\documentclass[preprint]{vgtc}               

\ifpdf
  \pdfoutput=1\relax                   
  \usepackage{graphicx}                
  \DeclareGraphicsExtensions{.pdf,.png,.jpg,.jpeg} 
\else
  \usepackage{graphicx}                
  \DeclareGraphicsExtensions{.eps}     
\fi%

\graphicspath{{figures/}{pictures/}{images/}{./}} 

\usepackage{microtype}                 
\PassOptionsToPackage{warn}{textcomp}  
\usepackage{textcomp}                  
\usepackage{mathptmx}                  
\usepackage{times}                     
\usepackage{cite}                      
\usepackage{tabu}                      
\usepackage{booktabs}                  
\usepackage[disable]{todonotes}
\newcommand{\florian}[1]{\todo[color=green!40,size=\scriptsize]{\textbf{F:} #1}}
\newcommand{\iflorian}[1]{\todo[inline,color=green!40]{\textbf{F:} #1}}

\onlineid{0}

\vgtccategory{Research}

\vgtcinsertpkg

\usepackage[utf8]{inputenc}
\usepackage{tikz}
\usepackage{amssymb}
\usepackage{amsmath}
\usepackage{mathtools}
\usepackage{subcaption}
\usepackage[linesnumbered,noend,ruled]{algorithm2e}
\usepackage{enumitem}
\usepackage{appendix}
\usepackage{amssymb}
\usepackage{pifont}
\usepackage{amsthm}
\usepackage{float}
\usepackage{doi}
\usepackage{pgfmath}
\usepackage{complexity}

\DeclareMathOperator{\troot}{root}

\DeclareMathOperator{\cost}{c}
\DeclareMathOperator{\pathstart}{\alpha}
\DeclareMathOperator{\pathend}{\omega}

\DeclareMathOperator{\desc}{desc}

\newcommand{\edist}{\delta_{\text{E}}}
\newcommand{\prunedvar}{\overline{d}}
\newcommand{\delvar}{d}
\newcommand{\ddelvar}{\hat{d}}
\newcommand{\parentvar}{p}
\newcommand{\pmvar}{pm}
\newcommand{\mapvar}{m}

\newcommand{\gurobi}{\textsc{Gurobi}}
\newcommand{\cplex}{\textsc{CPLEX}}
\newcommand{\scip}{\textsc{SCIP}}
\newcommand{\pulp}{\textsc{PuLP}}

\newtheorem{theorem}{Theorem}

\newtheorem*{definition*}{Definition}

\usetikzlibrary{decorations.pathreplacing,calligraphy}

\title{Taming Horizontal Instability in Merge Trees: On the Computation of a Comprehensive Deformation-based Edit Distance}

\author{Florian Wetzels\thanks{e-mail: wetzels@rptu.de}\\ %
        \scriptsize University of Kaiserslautern-Landau %
\and Markus Anders\\ %
     \scriptsize TU Darmstadt %
\and Christoph Garth\thanks{e-mail: garth@rptu.de}\\ %
     \scriptsize University of Kaiserslautern-Landau}

\teaser{
  \centering
  \includegraphics[angle=270,width=0.95\linewidth]{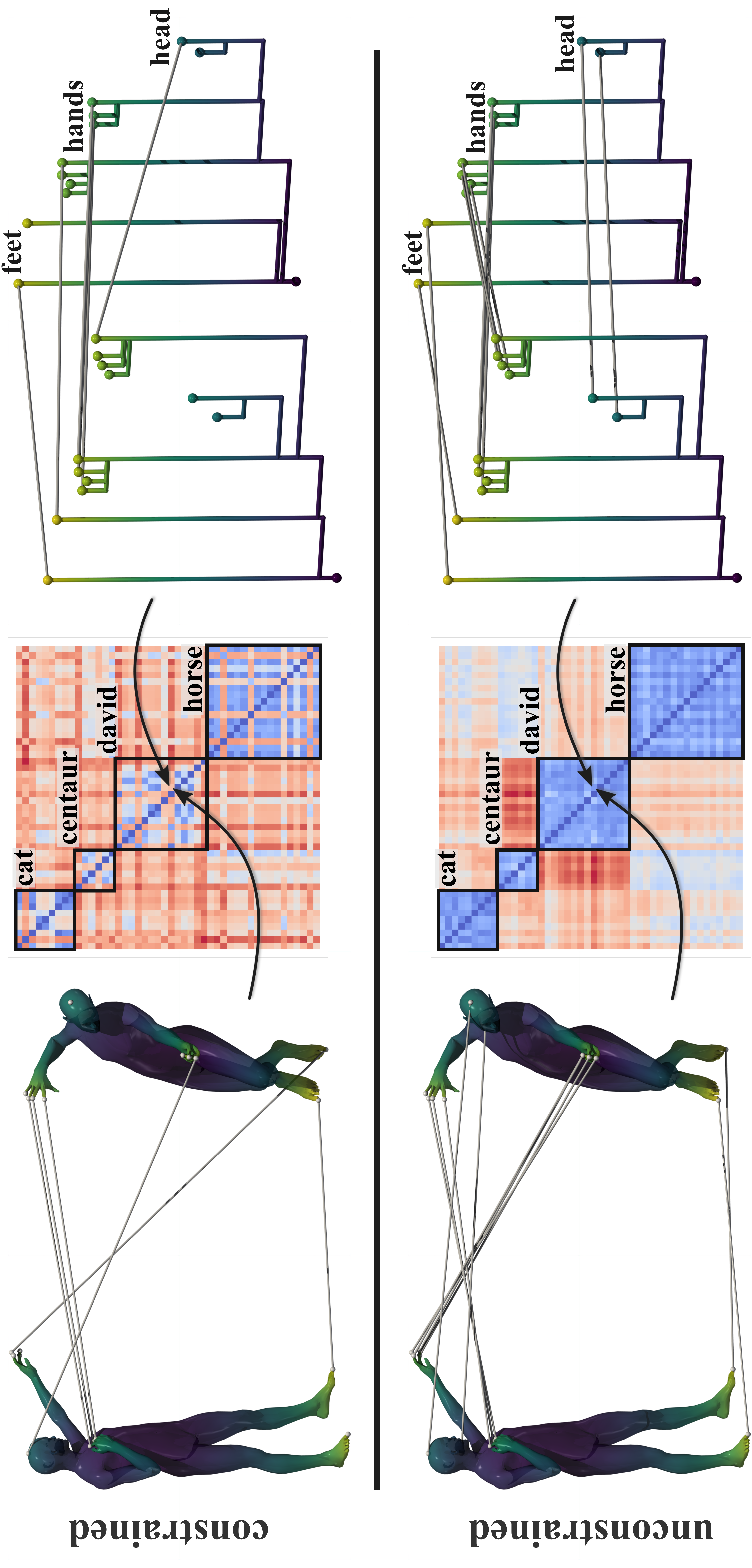}
  \caption{Illustration of the improvements of unconstrained edit distances over constrained edit distances on the TOSCA ensemble: two embedded mappings of critical points are shown on the left, the same mappings in a typical merge tree layout on the right. In between, distance matrices for multiple members of the ensemble can be found. The merge trees on the right contain a saddle swap: the nesting of head- and hand subtrees is changed between the first and the second tree. Such saddle swaps cause semantically poor mappings (left and right, top row) as well as cluster being overshadowed by noise in the distance matrix(center) when using constrained edit distances. Such problems do not arise for the unconstrained distance (bottom row). The example instance shown in the two mappings can also be observed to contribute noise within the ``david'' cluster in the upper distance matrix, but not in the lower one.}
  \label{fig:teaser}
}

\abstract{

Comparative analysis of scalar fields in scientific visualization often involves distance functions on topological abstractions.
This paper focuses on the merge tree abstraction (representing the nesting of sub- or superlevel sets) and proposes the application of the unconstrained deformation-based edit distance.
Previous approaches on merge trees often suffer from instability: small perturbations in the data can lead to large distances of the abstractions.
While some existing methods can handle so-called vertical instability, the unconstrained deformation-based edit distance addresses both vertical and horizontal instabilities, also called saddle swaps.
We establish the computational complexity as \NP-complete, and provide an integer linear program formulation for computation.
Experimental results on the TOSCA shape matching ensemble provide evidence for the stability of the proposed distance.
We thereby showcase the potential of handling saddle swaps for comparison of scalar fields through merge trees.
} 


\begin{document}



\firstsection{Introduction}
\maketitle

Comparative analysis of scalar fields is a core subject in the field of scientific visualization.
Over the last years, comparisons performed on topological abstractions have received increased interest, as it has two major advantages over direct comparisons between scalar fields:
First, abstract comparisons capture purely topological similarity, which for example reduces the impact of geometric symmetries.
Second, the abstractions are typically orders of magnitude smaller than the actual domain.
The latter aspect is of utmost interest in times of quickly increasing complexity of ensemble datasets.
One abstraction that has received particularly high interest is the merge tree, which represents the nesting of super- or sublevel sets in a rooted tree structure.

One possible approach is to apply tree edit distances to merge trees.
Tree edit distances are an established framework for measuring similarity of rooted trees~\cite{DBLP:journals/jacm/Tai79,DBLP:journals/ipl/ZhangSS92,treeEditSurvey,DBLP:journals/talg/BringmannGMW20}.
Typically, these metrics are intuitive, efficiently computable, and correspond to mappings between substructures.
Moreover, the metric property and edit mappings makes them suitable for tasks beyond simple distance computations, such as feature tracking, interpolation, or clustering~\cite{DBLP:journals/cgf/LohfinkWLWG20,DBLP:journals/tvcg/PontVDT22,wetzels2022branch}.

Specifically for merge trees, there is a rapid development of specialized tree edit distances~\cite{DBLP:journals/cgf/SaikiaSW14,DBLP:journals/tvcg/SridharamurthyM20,DBLP:journals/tvcg/SridharamurthyN23,DBLP:journals/tvcg/PontVDT22,wetzels2022branch,wetzels2022path} and their applications in various analysis tasks on scalar fields~\cite{DBLP:journals/cgf/SaikiaSW14,DBLP:journals/cgf/LohfinkWLWG20,DBLP:journals/tvcg/PontVDT22,DBLP:journals/tvcg/SridharamurthyN23}.
A major hurdle for these distances are so-called vertical and horizontal instabilities (using the notation from Saikia et al.~\cite{DBLP:journals/cgf/SaikiaSW14}).

Vertical instabilities stem from using an abstract representation instead of the merge tree itself: the persistence-based branch decomposition, derived by the so-called elder rule~\cite{edelsbrunner09} (we will omit details here).
Here, small-scale perturbations in the data can lead to a change in the hierarchy of branches, which in turn leads to poor-quality results of distances based on this hierarchy.
One of the most recent works~\cite{wetzels2022path} introduced the so-called deformation-based edit distance, which circumvents the use of branch decompositions.
The paper provided a novel notion of edit operations on merge trees as well as an algorithm for a constrained variant of this distance: the path mapping distance.
This path mapping distance is significantly less susceptible to vertical instabilities than previous methods working on persistence-based branch decompositions.

While there are existing approaches to tackling vertical instability, horizontal instabilities (sometimes also referred to as saddle swaps) remain a core problem and have been identified as such in several other works~\cite{DBLP:journals/cgf/SaikiaSW14,DBLP:journals/tvcg/SridharamurthyM20,DBLP:journals/tvcg/PontVDT22}.
Indeed, they have not been addressed by any of the existing edit distances.

In this paper, we study a merge tree edit distance that addresses both vertical and horizontal instabilities: the unconstrained deformation-based edit distance. 
This distance has been proposed but not studied in previous work~\cite{wetzels2022path}.

Our contribution is two-fold:
first, we establish the computational complexity of the unconstrained deformation-based edit distance.
It turns out that it is \NP-complete to compute the distance.
Secondly, we provide means to actually compute the distance.
We give a formulation of the problem as an integer linear program.
A collection of further optimizations enables feasible running times on merge trees of up to at least 25 vertices.
The source code of our implementation is provided in supplementary material, and we will release it in an open source repository~\cite{repository}.

We utilize this method to experimentally demonstrate the vertical and horizontal stability of this distance on two synthetic datasets as well as the well-known TOSCA shape matching ensemble (see Figure~\ref{fig:teaser}), though making heavy use of simplification.
In addition, we discuss why the distance should be expected to be more stable on a conceptual level.
A formal study of stability properties is left for future work.

After introducing required terminology in Section~\ref{sec:preliminaries}, we describe the IP formulation of the distance in Section~\ref{sec:computation}.
The proof of \NP-completeness is given in Section~\ref{sec:hardness}.
Section~\ref{sec:discussion} discusses how the proposed distance achieves stability.
Our experiments are presented in Section~\ref{sec:experiments} and Section~\ref{sec:conclusion} concludes the paper.
In the remainder of this section, we give an overview on related work.



\subsection*{Related Work}
Topological abstractions are ubiquitous in the comparison of scalar fields as well as in scientific visualization in general.
For a general introduction into topological methods for scalar field analysis, we refer to the survey by Heine et al.~\cite{heine16}.
A quite recent survey on scalar field \emph{comparisons} via topological descriptors was given by Yan et al.~\cite{surveyComparison2021}.
In this paper, we focus specifically on \emph{edit distances} between \emph{merge trees}.
A survey on edit distances between rooted trees can be found in~\cite{treeEditSurvey}.
The unconstrained edit distance studied in this paper is an adaptation of the general tree edit distances on unordered trees, which was introduced by Zhang~\cite{DBLP:journals/ipl/ZhangSS92}.

For our deformation-based distance, we use a reduction to integer linear programming that is based on the approach by Kondo et al.~\cite{DBLP:conf/dis/KondoOIY14}.
An improved method for the classic tree edit distance can be found in Hong et al.~\cite{DBLP:conf/cocoa/HongKY17}.

In the remainder of this section, we recap the most important methods for scalar field comparison through topological descriptors that are closest to our distance.

The first class of methods are other edit distances on merge trees.
Most noteworthy are the merge tree edit distance by Sridharamurthy et al.~\cite{DBLP:journals/tvcg/SridharamurthyM20}, the merge tree Wasserstein distance by Pont et al.~\cite{DBLP:journals/tvcg/PontVDT22}, the extended branch decomposition graph method by Saikia et al.~\cite{DBLP:journals/cgf/SaikiaSW14}, as well as the branch and path mapping distances by Wetzels et al.~\cite{wetzels2022branch,wetzels2022path}.
Regarding stability, none of these distances have been studied formally, which might be due to the nature of edit distances: they differ significantly from bottleneck distances which are typically used as the baseline, i.e.\ in relation to which stability is usually defined.
However, the three branch decomposition-based methods are known to be susceptible against vertical instabilities~\cite{DBLP:journals/cgf/SaikiaSW14,wetzels2022branch}.
In contrast, vertical stability was observed experimentally for the branch decomposition-independent methods by Wetzels et al.~\cite{wetzels2022branch,wetzels2022path}.
There is more work on merge tree edit distances which focuses rather on certain applications, instead of the actual distance measure~\cite{DBLP:journals/cgf/LohfinkWLWG20,DBLP:journals/tvcg/PontVDT22,DBLP:journals/tvcg/SridharamurthyN23}.

Most other distances proposed for topological descriptors are defined in the vein of bottleneck distances.
Bottleneck distances focus on measuring the largest change, instead of summing up all the changes like edit distances do.
Such methods exist for merge trees~\cite{BeketayevYMWH14,morozov14,DBLP:journals/tvcg/BollenTL23}, persistence diagrams~\cite{interleaving_distance,Cohen-Steiner2007,edelsbrunner09}, and Reeb graphs~\cite{DBLP:conf/3dor/BauerFL16,DBLP:journals/dcg/FabioL16}.
Stability properties of these distances have been studied more extensively than for edit distances.
However, none of these distances is known to be stable, discriminative, and efficiently computable at the same time~\cite{DBLP:journals/tvcg/BollenTL23}.
Moreover, most of them also lack a publicly available implementation.

Lastly, there are alternative distance measures which focus on combining topological and geometrical similarity~\cite{YanWMGW20,Yan_geometry_aware,intrinsicMTdistance,ThomasN13,DBLP:conf/apvis/NarayananTN15}.

We should note that the closest work to ours is probably the paper by Bollen et al.~\cite{DBLP:journals/tvcg/BollenTL23}: it aims to define a stable distance for merge trees, investigates the stability experimentally and derives an exponential time algorithm.
Our work differs in the type of mapping considered: we study an edit distance, whereas the distance by Bollen et al.\ is similar to a bottleneck distance or the works by Morozov et al.~\cite{morozov14} or Beketayev et al.~\cite{BeketayevYMWH14}.
The mappings considered by such distances are inherently different from edit mappings.
As stated above, edit mappings are very powerful tools to derive further visualization or analysis methods.
Moreover, we provide complexity results for the studied edit distance.

Our technique has a similar goal as the application of so-called $\epsilon$-preprocessing~\cite{DBLP:journals/tvcg/SridharamurthyM20}.
This technique simplifies the input merge trees prior to the distance computation.
It collapses short inner edges and thereby merges close saddles.
Thus, it tries to reduce the number of instabilities that appear in the input, instead of handling them through the distance.
Examples for algorithms which make use of such a preprocessing step can be found in~\cite{DBLP:journals/tvcg/SridharamurthyM20,DBLP:journals/tvcg/PontVDT22}.

While this technique fixes instabilities in some instances, it has two drawbacks.
For small values of $\epsilon$, the distance still suffers from instability: due to the fixed threshold value, larger instabilities remain in the data.
For large values, the preprocessing removes possibly important subtleties in the data, such that the distance applied afterwards can no longer reflect these details.
Clearly, this implies that the preprocessing is not a general solution to the problem, although it can yield significant improvements on certain datasets.

This stands in contrast to the unconstrained deformation-based edit distance, which we discuss in Section~\ref{sec:discussion}.
However, unconstrained edit distances introduce a significant increase in complexity whereas the preprocessing can be done in linear time.
We experimentally compare our method to the effect of $\epsilon$-preprocessing with different threshold values in Section~\ref{sec:experiments}.

\section{Preliminaries}
\label{sec:preliminaries}
In this section, we provide the formal background for this paper.
First, we define abstract merge trees, the class of trees studied throughout the paper.
Then, we introduce the unconstrained deformation-based edit distance, which was first proposed in~\cite{wetzels2022path}.
Finally, we give a short recap on integer linear programming, since we use it to compute our edit distance in practice.

\subsection*{Merge Trees}
We begin with a definition of merge trees.
In this paper, we only consider abstract merge trees, which represent exactly those graphs that can be interpreted as a merge tree for some domain.
For a detailed definition of merge trees, see~\cite{DBLP:conf/focs/EdelsbrunnerLZ00} or~\cite{DBLP:conf/ppopp/MorozovW13}.
We restate definitions from~\cite{wetzels2022branch,wetzels2022path} for the basic graph theoretic and topological concepts.
For an introduction into basic notions from computational topology and topological data analysis, we refer to~\cite{edelsbrunner09} or~\cite{heine16}. 

We consider rooted trees as directed graphs with parent edges.
In particular, a rooted tree $T$ is a directed graph with vertex set $V(T)$, edge set $E(T) \subseteq V(T) \times V(T)$ and a unique root, denoted $\troot(T)$.
We call a node $c \in V(T)$ a child of node $p \in V(T)$, if $(c,p) \in E(T)$, and, conversely, $p$ the parent of $c$.
For a node $p$, we denote its number of children by $\deg_T(p) \coloneqq |\{ c \mid (c,p) \in E(T) \}|$ and all its descendants by $\desc_T(p)$.
We omit the index $T$ when it is clear from the context.
Furthermore, we denote the empty tree, consisting only of a single node and no edges, by $\bot$.

Merge trees inherit the scalar function from their original domain and thus are labeled trees.
Usually, they are considered as node-labeled trees, where the nodes are the critical points of the scalar field and the labels being defined through the original scalar values at these critical points.
Those rooted trees that can be interpreted as merge trees for some domain of dimension at least~$2$ are called \emph{abstract} merge trees. 
These objects are the main focus of this paper.

\begin{definition*}[Abstract Merge Tree]
An unordered, rooted tree $T$ of (in general) arbitrary degree with node labels $f:V(T) \rightarrow \mathbb{R}$ is an \emph{abstract merge tree} if the following properties hold:
\begin{itemize}
    \item the root node has degree one, $\deg( \troot (T) ) = 1$,
    \item all inner nodes have a degree of at least two,\\ $\deg(v) \neq 1$ for all $v \in V(T)$ with $v \neq \troot (T)$,
    \item all nodes have a larger scalar value than their parent node, $f(c) > f(p)$ for all $(c,p) \in E(T)$.
\end{itemize}
\end{definition*}
\noindent
The deformation-based edit distance, as defined in~\cite{wetzels2022path}, works on edge-labeled trees instead.
Here, edge labels represent the length of the scalar range of the edge (we also say its \emph{persistence}).
However, these two representations are interchangeable:
given a node label function $f : V(T) \rightarrow \mathbb{R}_{>0}$, we define the corresponding edge label function $\ell_f : E(T) \rightarrow \mathbb{R}_{>0}$ by $\ell_f((u,v)) = |f(u) - f(v)|$.
Given an edge label function, we can again define $f_\ell$ by placing the root node at a fixed scalar value, e.g.~$0$.
We denote the total persistence (the sum of all edge-lengths) of an abstract merge tree $T$ by \mbox{$||T|| := \sum_{e \in E(T)} \ell(e)$}.


As for general graphs, a \emph{path} of length $k$ in an abstract merge tree $T$ is a sequence of vertices $p=v_1 ... v_k \in V(T)^k$ with $(v_{i},v_{i-1}) \in E(T)$ for all $2 \leq i \leq k$ and $v_i \neq v_j$ for all $1 \leq i,j \leq k$. Note that we only consider paths in root-to-leaf direction.
For a path $p=v_1 ... v_k$, we denote its first vertex by $\pathstart(p) \coloneqq v_1$, its last vertex by $\pathend(p) \coloneqq v_k$ and the set of all paths of a tree $T$ by $\mathcal{P}(T)$.


We lift the label function $\ell$ of an abstract merge tree $T$ from edges to paths in the following way: $\ell(v_1...v_k) = \sum_{2 \leq i \leq k} \ell((v_i,v_{i-1}))$.

\subsection*{Edit Operations}

We now recap the definition of the deformation-based edit distance which was introduced in~\cite{wetzels2022path}.
It is based on the well-established edit distance on unordered trees by Zhang~\cite{DBLP:journals/ipl/ZhangSS92}, but uses an adapted set of edit operations tailored to merge trees.
The distance considers sequences of the following three edit operations which transform one tree into another: the relabel operation changes the length or label of an edge, the deletion contracts an edge completely, and the insertion adds a new edge to the tree.

In contrast to classic tree edit distance, after a deletion, we prune remaining nodes of degree one by merging its two incident edges (and their lengths).
Figure~\ref{fig:continuousDelete} illustrates the intuition behind this specific deletion.
Insertions are then defined to be the inverse operations of these deletions.

\begin{figure}[]
  \centering
  \resizebox{\linewidth}{!}{
  \begin{tikzpicture}[yscale=0.5]
  
  \node[draw,circle,fill=gray!100,minimum width=0.5cm] at (0, 0) (root_1) {};
  \node[draw,circle,fill=gray!100,minimum width=0.5cm] at (0, 3) (s1_1) {};
  \node[draw,circle,fill=red!80,minimum width=0.5cm] at (-2, 12) (m1_1) {};
  \node[draw,circle,fill=red!80,minimum width=0.5cm] at (2, 12) (m2_1) {};
  \node[draw,circle,fill=gray!100,minimum width=0.5cm] at (1.333, 9) (s2_1) {};
  \node[draw,circle,fill=red!80,minimum width=0.5cm] at (0.3, 11) (m3_1) {};
  \draw[gray,very thick] (root_1) -- (s1_1);
  \draw[gray,very thick] (s1_1) -- (m1_1);
  \draw[gray,very thick] (s1_1) -- (s2_1);
  \draw[gray,very thick] (s2_1) -- (m2_1);
  \draw[gray,very thick] (s2_1) -- (m3_1);
  
  \node[draw,circle,fill=gray!100,minimum width=0.5cm] at (0+8, 0) (root_2) {};
  \node[draw,circle,fill=gray!100,minimum width=0.5cm] at (0+8, 3) (s1_2) {};
  \node[draw,circle,fill=red!80,minimum width=0.5cm] at (-2+8, 12) (m1_2) {};
  \node[draw,circle,fill=red!80,minimum width=0.5cm] at (2+8, 12) (m2_2) {};
  \draw[gray,very thick] (root_2) -- (s1_2);
  \draw[gray,very thick] (s1_2) -- (m1_2);
  \draw[gray,very thick] (s1_2) -- (m2_2);
  
  \node[draw,circle,fill=gray!100,minimum width=0.5cm] at (0-8, 0) (root_3) {};
  \node[draw,circle,fill=gray!100,minimum width=0.5cm] at (0-8, 3) (s1_3) {};
  \node[draw,circle,fill=red!80,minimum width=0.5cm] at (-2-8, 12) (m1_3) {};
  \node[draw,circle,fill=red!80,minimum width=0.5cm] at (2-8, 12) (m2_3) {};
  \node[draw,circle,fill=gray!100,minimum width=0.5cm] at (1.333-8, 9) (s2_3) {};
  \draw[gray,very thick] (root_3) -- (s1_3);
  \draw[gray,very thick] (s1_3) -- (m1_3);
  \draw[gray,very thick] (s1_3) -- (s2_3);
  \draw[gray,very thick] (s2_3) -- (m2_3);
  
  \draw[-latex,ultra thick] (-1.5,6) to[bend left=20] (-6.5,6) node [midway,above] {};
  \draw[-latex,ultra thick] (1.5,6) to[bend right=20] (6.5,6) node [midway,above] {};
  \node[] at (-4, 6.7) (label_1) {\huge classic delete};
  \node[] at (4, 6.7) (label_2) {\huge continuous delete};
  
  \end{tikzpicture}
  }
  \caption{Intuition behind the deformation-based edit operations: if a degree-1 node remains after a deletion, we obtain an invalid merge tree. To fix this problem, we also have to prune the remaining node.}
  \label{fig:continuousDelete}
\end{figure}
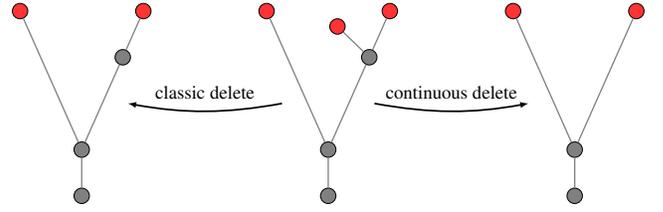

Formally, we consider the following edit operations that transform an abstract merge tree $T,\ell$ into another abstract merge tree $T',\ell'$:
\begin{itemize}
    \item Edge relabel: changing the length of an edge $(c,p)$ to a new value $v \in \mathbb{R}_{>0}$, i.e.\ $T'=T$, $\ell'((c,p)) = v$ and $\ell'(e)=\ell(e)$ for all $e \neq (c,p)$.
    \item Edge contraction: remove an edge from the tree and merge the two nodes. Then, remove the parent node if it had only two children originally. Formally, for a node $p$ with children $c_0...c_k$ and parent $p'$, we define $T'$ after contracting $(c_i,p)$ as follows: if $k>1$, we have
    $$V(T') = V(T) \setminus \{c_i\}, E(T') = E(T) \setminus \{(c_i,p)\},$$
    and otherwise, if $k=1$, we have
    $$V(T') = V(T) \setminus \{c_i,p\},$$
    $$E(T') = (E(T) \cup \{(c_{1-i},p')\}) \setminus \{(c_i,p),(c_{1-i},p),(p,p')\}.$$
    Furthermore, $\ell' = \ell$ if $k>1$, otherwise
    $$\ell'((c_{1-i},p')) = \ell((p,p'))+\ell((c_{1-i},p)),$$ and $\ell'(e)=\ell(e)$ for all $e \neq (c_{1-i},p')$.
    \item Inverse edge contraction: inverse operation to edge contraction.
\end{itemize}
We use the terms edge contractions and deletions interchangeably as well as inverse edge contractions and insertions.
If a sequence of edit operations $S$ transforms an abstract merge tree $T_1$ into $T_2$, we denote this by $T_1 \xrightarrow{\scriptscriptstyle S} T_2$.

We define the cost of an edit operation as the euclidean distance on $\mathbb{R}_{>0} \cup \{0\}$: $ \cost(l_1,l_2) = |l_1-l_2| $ for all $l_1,l_2 \in \mathbb{R}_{\geq 0}$.
Here, an edit operation is represented by the labels of the modified edges: a relabel of an edge $e$ with $(\ell(e),\ell'(e))$, whereas insertions or deletion of $e$ are represented by $(0,\ell(e))$ or $(\ell(e),0)$.
This means, for a deletion or insertion we charge the persistence of the edge, and for a relabel we charge the persistence difference between the old and new edge.

The cost of an edit sequence is then the sum of all edit operation costs: $ \cost(s_1 \dots s_k) = \sum_{1 \leq i \leq k} |\cost(s_i)| $.
Based on this, the unconstrained deformation-based edit distance between two trees $T_1,T_2$ is defined as the cost of a cost-minimal sequence that transforms $T_1$ into $T_2$:
$$ \edist(T_1,T_2) = \min\{ \cost(S) \mid T_1 \xrightarrow{\scriptscriptstyle S} T_2 \}. $$
This differs from its constrained variant: the path mapping distance is defined as the cost of an optimal sequence only using deletions and insertions on leafs.
In this paper, we focus on the computation and complexity of the unconstrained variant, in contrast to prior work~\cite{wetzels2022path} only studying the constrained path mapping distance.

\subsection*{Integer Linear Programming}
We use integer linear programming (IP) to compute our edit distance. 
An IP instance consists of a number of variables, an objective function, and a set of (linear) constraints. The goal is to minimize the objective function, while the constraints must be satisfied.

Formally, we expect an IP instance to be in the following form:  
\begin{equation*}
  \begin{array}{ll@{}ll}
  \text{minimize}  & \displaystyle c^T x& &\\
  \text{subject to}& Ax \geq b&\\
                   & x \geq 0\\
                   & x \in \mathbb{Z}^n
  \end{array}
  \end{equation*}
where $x$ is the vector of $n$ variables, $\mathbf{c} \in \mathbb{R}^n$ the weights used for the objective function, and $A \in \mathbb{R}^{n \times m}$ a matrix used to denote constraints. 
Note that the variables are constrained to be integers.

The problem is well-known to be \NP-hard.
Even still, integer programming is often used to express and solve other \NP-hard optimization problems in practice, for example problems stemming from logistics or combinatorics. 
There are powerful solvers available, such as \gurobi{} \cite{gurobi}, \cplex{} \cite{cplex2009v12}, or \scip{} \cite{BestuzhevaEtal2021OO}.

\section{Computation}
\label{sec:computation}

Algorithms for tractable variants of edit distances usually compute the optimal \emph{edit mapping} (explained later in more detail) between the two trees by exploiting their recursive structure using dynamic programming techniques.
This holds for edit distances on general trees (e.g.~\cite{DBLP:journals/algorithmica/Zhang96}) as well as merge tree-specific methods (e.g.~\cite{wetzels2022path}).
For non-tractable variants such as the general tree edit distance on unordered trees, it is more common to compute the optimal mapping through a reduction to integer programming~\cite{DBLP:conf/dis/KondoOIY14} or combinations of recursive decompositions and linear optimizations, see~\cite{DBLP:conf/cocoa/HongKY17}.

Our approach is to adapt the conventional reduction to integer programming, which we describe in the following.

\subsection*{Edit Mappings}

We first recap the concept of edit mappings~\cite{DBLP:journals/ipl/ZhangSS92,treeEditSurvey}.
Edit mappings represent edit sequences between two trees in an abstract way.
Each edit mapping corresponds to an edit sequence.
In terms of the resulting distance, the optimal edit mapping and optimal edit sequence are guaranteed to be equivalent. 

Let $T_1,T_2$ be two trees with $T_1 \xrightarrow{\scriptscriptstyle S} T_2$.
If a node $v \in V(T_1)$ is not touched by any operation in $S$ or only by relabels, then it has a clear correspondent $v' \in V(T_2)$.
Thus, $(v,v')$ is part of the edit mapping corresponding to $S$.
If a node is deleted or inserted by $S$, it is not present in the mapping.
The way edit operations on trees work induces the following properties of an edit mapping $M \subseteq V(T_1) \times V(T_2)$ between the nodes of two trees:
(a) They are one-to-one mappings, i.e.\ for all $(v_1,v_2) \in M$ and $(u_1,u_2) \in M$ we have $v_1 = u_1$ if and only if $v_2 = u_2$.
(b) They are ancestor preserving, i.e.\ for all $(v_1,v_2) \in M$ and $(u_1,u_2) \in M$ we have $v_1$ is an ancestor of $u_1$ if and only if $v_2$ is an ancestor of $u_2$.

As mentioned above, the cost of an optimal mapping (we omit the exact definition) is exactly the cost of an optimal edit sequence.
Hence, we may consider edit mappings instead of edit sequences when computing the edit distance.
In fact, practical algorithms usually iterate over edit mappings instead of edit sequences.
In the following, we also make use of edit mappings.

As discussed earlier, for our deformation-based edit distance, the edit operations differ from the typical edit distances.
Hence, we also need to adapt the concept of edit mappings accordingly.
We extend the concept of the edit mappings above.
More specifically, we use the conventional edit mappings on nodes (as defined above) as our foundation, but redefine the cost function for a given edit mapping to reflect precisely our deformation-based distance.


Given two trees $T_1,T_2$ and an ancestor-preserving node mapping $M$ between them, we now define the costs $\cost(M)$ of $M$.
We begin by defining a few variables for each node in order to ease the process.
Given a node $v \in V(T_1)$, we first define $\delvar_v = 1$ whenever $v \notin M$, and $\delvar_v = 0$ otherwise.
Thus, $\delvar_v$ denotes if $v$ is deleted in $M$.
Based on this, we define $\ddelvar_v = \bigwedge_{u \in \desc(v)} \delvar_u$. 
Hence, $\ddelvar_v$ is $1$ whenever the whole subtree rooted in $v$ is deleted, and $0$ otherwise.

Next, we define so-called \emph{pruned} nodes.
Recall the definition of a continuous deletion on a merge tree:
if a node of degree $1$ remains after a deletion, we replace its up- and down-edge with one merged edge and remove the node.
We call this \emph{pruning} of a node (see Figure~\ref{fig:continuousDelete}).
The merged edge in the resulting tree corresponds to a path in the tree prior to the deletion.
Thus, the start and end vertices of this path are in an imaginary parent-child relation.

We now define variables which express pruned nodes.
Indeed, we observe that a node is pruned if and only if all but one of its subtrees are completely deleted.
Thus, we may define the pruning variable $\prunedvar_v$ for a node $v$ with children $c_1, c_2, \dots c_{\deg(v)}$ to be $1$ if and only if $\sum_{1 \leq i \leq \deg(v)} (1-\ddelvar_{c_i}) = 1$.

Next, we define variables representing the imaginary parent for each node in the two trees, deriving them from the deleted nodes and pruned nodes.
Intuitively, the imaginary parent of a node is its least ancestor that is not pruned.
An additional constraint is that we only define an imaginary parent for nodes that are not deleted.
Formally, we define a variable $\parentvar_{u,v}$ for each pair of nodes from the same tree.
Then, $\parentvar_{u,v} = 1$ if and only if $\delvar_v = 0$ and $\prunedvar_w = 1$ for each node $w$ between $v$ and $u$.

As a last step, we can derive the relabel operations of the mapping $M$.
They are represented by mapped nodes together with their imaginary parents.
All other edges in the trees that are not covered by these relabels can be assumed to be deleted.
The costs of the mapping $M$ are then defined to be
$$ \sum_{q_1,q_2 \in \mathcal{P}_1 \times \mathcal{P}_2} \parentvar_{\pathstart_1, \pathend_1} \cdot \parentvar_{\pathstart_2, \pathend_2,} \cdot (\cost(q_1,q_2)) $$
$$ - \sum_{\ell(q_1),\ell(q_2) \in \mathcal{P}_1 \times \mathcal{P}_2} \parentvar_{\pathstart_1, \pathend_1} \cdot \parentvar_{\pathstart_2, \pathend_2} \cdot  (\cost(\ell(q_1),0) + \cost(0,\ell(q_2))) $$
$$ + \sum_{(u_1,v_1) \in E(T_1)} \cost(\ell(u_1v_1),0) + \sum_{(u_2,v_2) \in E(T_2)} \cost(0,\ell(u_2v_2)) $$
where $\mathcal{P}_1 = \mathcal{P}(T_1)$, $\mathcal{P}_2 = \mathcal{P}(T_2)$, $\pathstart_1 = \pathstart(q_1)$, $\pathstart_2 = \pathstart(q_2)$, $\pathend_1 = \pathend(q_1)$ and $\pathend_2 = \pathend(q_2)$.

\subsection*{IP Definition}
For the IP instance, we first define variables for the basic mapping $M$.
For each pair of nodes $u,v \in V(T_1) \times V(T_2)$, we introduce a variable $\mapvar_{u,v}$ with $\mapvar_{u,v} = 1$ if and only if $(u,v) \in M$.
From the definition of edit mappings, we derive the same constraints on the mapping variables as in~\cite{DBLP:conf/dis/KondoOIY14}:
\begin{itemize}
  \item $\sum_{v \in V(T_2)} \mapvar_{u,v} \leq 1$ for all $u \in V(T_1)$ and $\sum_{u \in V(T_1)} \mapvar_{u,v} \leq 1$ for all $v \in V(T_2)$,
  \item $\mapvar_{u,v} + \mapvar_{x,y} \leq 1$ for all tuples $u,v,x,y$ contracting the ancestor preservation.
\end{itemize}

Next, we add all the variables and constraints based on the previous section.
This in turn expresses the special cost function for the deformation-based distance, which is in turn used as the objective for the integer program.    
We constrain all variables to $\{0,1\}$ values. 
However, note that our cost function is actually \emph{not} yet linear: we are multiplying $\parentvar_{\pathstart_1, \pathend_1} \cdot \parentvar_{\pathstart_2, \pathend_2}$.
Fortunately, we only need to compute $\parentvar_{\pathstart_1, \pathend_1} \wedge \parentvar_{\pathstart_2, \pathend_2}$ where $\parentvar_{\pathstart_1, \pathend_1}$ and $\parentvar_{\pathstart_2, \pathend_2}$ are treated as boolean variables. 
This is possible in integer programs using the introduction of an additional variable for said operation: we set $\pmvar_{\pathend_1,\pathstart_1,\pathend_2,\pathstart_2} := \parentvar_{\pathstart_1, \pathend_1} \wedge \parentvar_{\pathstart_2, \pathend_2}$.

We observe that the total number of variables is dominated by the
$\pmvar_{\pathstart_1,\pathend_1,\pathstart_2,\pathend_2}$
variables.
Hence, in the following, where we discuss optimizations, we focus on reducing the number of $\pmvar_{\pathstart_1,\pathend_1,\pathstart_2,\pathend_2}$ variables. 

\subsection*{Optimizations: Reencoding}
The encoding as described above turns out to be quite slow, even for very small examples. 
Our main strategy for optimization is to use a \emph{reencoding} scheme (loosely inspired by \cite{DBLP:conf/hvc/MantheyHB12}). 

Reencoding means we run \gurobi{} for a while, extract an upper bound for the edit distance (by using the best solution computed up to a certain point), and then reencode the problem more concisely using the newly found upper bound.
Generally, we perform the scheme using an exponential backoff strategy: we start with a time limit of $10$s for \gurobi{}, doubling the time limit in each iteration. 
After each iteration, if an improved upper bound is found, we reencode the problem.

The new encoding can be more compact as a given upper bound lets us discard potential edit mappings:
if a single mapping of nodes (or paths) already implies costs greater than the current upper bound, we know that this mapping can not occur in any optimal solution.  

This matches precisely a single $\pmvar_{\pathstart_1,\pathend_1,\pathstart_2,\pathend_2}$ variable: if we can deem the mapping of the path $(\pathstart_1 \dots \pathend_1)$ to $(\pathstart_2 \dots \pathend_2)$ too expensive (above the current upper bound), we can remove the variable $\pmvar_{\pathstart_1,\pathend_1,\pathstart_2,\pathend_2}$, all corresponding constraints, as well as variables which were only used in correspondence to $\pmvar_{\pathstart_1,\pathend_1,\pathstart_2,\pathend_2}$.

The question that therefore remains is, how can we approximate the cost implied by a single mapping of paths $\pmvar_{\pathstart_1,\pathend_1,\pathstart_2,\pathend_2}$? More specifically, we need to be able to give a lower bound for these costs.
Given two trees $T_1, T_2$ and a $\pmvar_{\pathstart_1,\pathend_1,\pathstart_2,\pathend_2}$ variable, we provide a collection of cost approximations which we can derive:
\begin{enumerate}
\item  The cost of mapping the paths themselves: the absolute difference between the persistence of the first path and the second path, i.e.\ $\cost(\ell(\pathstart_1 \dots \pathend_1),\ell(\pathstart_2 \dots \pathend_2))$.
\item If $(\pathstart \dots \pathend)$ is a path, all nodes on the path (strictly inbetween $\pathstart$ and $\pathend$) must be pruned. This means all subtrees of these nodes, which are not part of the path must be deleted, and the cost of these deletions can be assumed.
\item Due to the ancestor relationship, the subtrees rooted in $\pathend_1$ and $\pathend_2$ must be matched onto each other. 
There are several ways to compute lower bounds for this: let $T_1(\pathend_1)$ and $T_2(\pathend_2)$ denote the subtrees rooted in $\pathend_1$ and $\pathend_2$, respectively. 
A valid lower bound is then $\edist(T_1(\pathend_1), T_2(\pathend_2))$. 
We compute this using our IP strategy recursively (using a time limit). 
However, before we do so, we first approximate the cost by computing the absolute difference between the total persistence of $T_1(\pathend_1)$ and $T_2(\pathend_2)$. 
Only if this crude approximation (in conjunction with the other cost approximations) does not yet suffice to exclude the given variable, we fall back to the recursive IP approach.
\item Due to the ancestor relationship, the overall trees up to $\pathstart_1$ and $\pathstart_2$, removing the trees below $\pathstart_1$ and $\pathstart_2$, must be matched onto each other. 
Again, we first use the persistence difference as a crude approximation first, and then rely on the recursive IP approach.
\end{enumerate}
Whenever the sum of these costs exceeds the current upper bound for a variable $\pmvar_{\pathstart_1,\pathend_1,\pathstart_2,\pathend_2}$, we remove the variable. 
All in all, for most instances in our experiments, this substantially reduces the number of possible mappings.

\subsection*{Optimizations: Symmetry Reduction}
We also reduce a number of symmetries in the IP formulation.
For a given variable $\pmvar_{\pathstart_1,\pathend_1,\pathstart_2,\pathend_2}$, assume $\pathend_1$ is a leaf while $\pathend_2$ is not. 
In this case, no node below $\pathend_2$ can be mapped to any node of the first tree, due to the ancestor relationship.
Hence, all nodes below $\pathend_2$ need to be deleted.

Instead, let us choose \emph{any} leaf below $\pathend_2$, say, $\pathend_l$.
Now since in the case where we use $\pathend_2$ we need to remove the entire subtree -- including $\pathend_l$ -- anyway, using $\pathend_l$ instead of $\pathend_2$ can indeed not yield a more expensive edit mapping.
Hence, we only consider those $\pmvar_{\pathstart_1,\pathend_1,\pathstart_2,\pathend_2}$ where $\pathend_1$ is a leaf if and only if $\pathend_2$ is a leaf.

Interestingly, in our testing, this optimization sped up \gurobi{} considerably, while it seemed to marginally slow down \cplex{}.

We also apply a similar reduction for the roots of the trees.
We observe that the only degree-1 node in a merge tree is its root.
Hence, after deleting it, we have to construct a new degree-1 root again, since the tree resulting from an edit sequence is a merge tree (after all, we are only comparing merge trees).
Thus, with similar arguments as in the leaf case, enforcing the roots to be mapped onto each does not yield a worse edit mapping.
We disregard all tuples $\pmvar_{\pathstart_1,\pathend_1,\pathstart_2,\pathend_2}$ where $\pathstart_1$ is the root of its tree while $\pathstart_2$ is not (or vice versa).

\section{NP-completeness}
\label{sec:hardness}
We now prove \NP-completeness of the decision variant of computing the unconstrained deformation distance.
The decision variant asks whether the distance $\edist$ between two trees is below a given cost threshold $c$.
Our reduction is an adaptation of the reduction for the edit distance on unordered trees by Zhang~\cite{DBLP:journals/ipl/ZhangSS92}. 
\begin{theorem}
  Deciding whether $\edist(T_1,T_2)$ of two abstract merge trees $T_1,T_2$ is below a given threshold $c$ is \NP-complete. 
\end{theorem}
\begin{proof}
We prove the claim by a reduction from exact covering by 3-sets (X3C) to computing the deformation-based edit distance. We begin by recalling the definition of the X3C problem. 

\textit{(Definition of X3C.)}
We define a universe $U := \{u_1, u_2, \dots{}, u_m\}$ of $m=3k$ elements, and a collection of subsets of these elements $S := \{S_1, S_2, \dots{}, S_n\}$.
Each $S_i \in S$ has precisely $3$ elements, i.e.\ $S_i := \{u_{i,1}, u_{i,2}, u_{i,3}\}$.
Furthermore, every element $u \in U$ is contained in at most $3$ sets of $S$.
Given a universe $U$ and collection of sets $S$, the goal is to check whether there is a $C \subseteq S$, where $C$ is a \emph{disjoint} collection of sets and each element in $U$ occurs in exactly one member of $C$.
The X3C problem is known to be \NP-complete~\cite{DBLP:books/fm/GareyJ79}.

\textit{(Reduction to X3C.)}
We now reduce X3C to our problem.
Given an X3C instance with universe $U$ and sets $S$, we construct the trees $T_1$ and $T_2$ as follows (the construction is illustrated in Figure~\ref{fig:np_reduction}).
We begin by defining an element gadget $G_i$ for each element $u_i \in U$, which are illustrated in Figure~\ref{fig:label_gadget}.
Each gadget is derived from the base gadget $G$ which consists of $m$ edges $e_1 , \dots , e_m$ attached to its root $r$ with the length of $e_i$ being $2i$.
Then, in $G_i$, we add another edge $e_0$ of length $1$ which splits $e_i$ exactly in halves.
This leads to the following property: the distance between each $G_i$ and $G$ is exactly $1$, since we simply have to contract the edges $e_0$ in $G_i$.
Furthermore, for $i \neq j$, the distance between $G_i$ and $G_j$ is exactly $2$, since we need to delete $e_0$ in one tree and insert it in the other.

\begin{figure}
  \centering
  \begin{tikzpicture}[scale=0.8,every node/.style={minimum size=0.2cm,inner sep=0,fill=gray!100,draw=black,circle},yscale=-1]
    \node at (0,0) (r) {};
    \node[fill=red!80] at (-2,-1) (a1) {};
    \node[fill=red!80] at (-1,-2) (a2) {};
    \node[draw=white!0,fill=white!0] at (0,-1.5) {$\dots{}$};
    \node[fill=red!80] at (2,-3)  (a3) {};

    \draw[] (r) -- (a1) node[midway, fill=white,draw=none,inner sep = 0.5] {$2$};
    \draw[] (r) -- (a2) node[midway, fill=white,draw=none,inner sep = 0.5] {$4$};
    \draw[] (r) -- (a3) node[midway, fill=white,draw=none,inner sep = 0.5] {$2m$};

    \node[draw=white!0,fill=white!0]  at (0,-3.5) (g) {\color{white}$G_i$};
    \node[draw=white!0,fill=white!0]  at (0,-3.5) (g) {$G$};
  \end{tikzpicture} \hspace{0.5cm}
  \begin{tikzpicture}[scale=0.8,every node/.style={minimum size=0.2cm,inner sep=0,fill=gray!100,draw=black,circle},yscale=-1]
    \node at (0,0) (r) {};
    \node[fill=red!80] at (-2,-1) (a1) {};
    \node at (0,-1.25) (ai) {};
    \node[fill=red!80] at (0,-2.5) (aii) {};
    \node[fill=red!80] at (-0.75,-1.75) (aii1) {};
    \node[draw=white!0,fill=white!0] at (-1,-1) {$\dots{}$};
    \node[draw=white!0,fill=white!0] at (1,-2.5) {$\dots{}$};
    \node[fill=red!80] at (2,-3)  (a3) {};

    \draw[] (r) -- (a1) node[midway, fill=white,draw=none,inner sep = 0.5] {$2$};
    \draw[] (r) -- (ai) node[midway, fill=white,draw=none,inner sep = 1.5] {$i$};
    \draw[] (ai) -- (aii) node[midway, fill=white,draw=none,inner sep = 1.5] {$i$};
    \draw[] (ai) -- (aii1) node[midway, fill=white,draw=none,inner sep = 0] {$1$};
    \draw[] (r) -- (a3) node[midway, fill=white,draw=none,inner sep = 0.5] {$2m$};

    \node[draw=white!0,fill=white!0]  at (0,-3.5) (g) {$G_i$};
  \end{tikzpicture}
  \caption{The element gadgets $G$ and $G_i$.} \label{fig:label_gadget}
\end{figure}
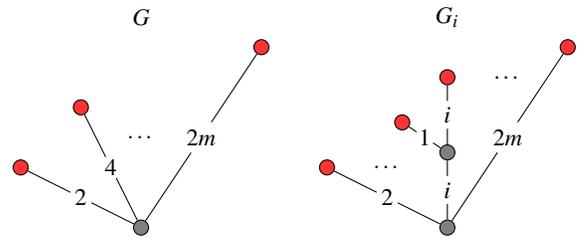

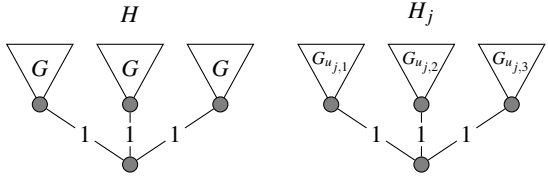
\begin{figure}
  \centering
  \begin{tikzpicture}[scale=0.8,every node/.style={minimum size=0.2cm,inner sep=0,fill=gray!100,draw=black,circle},yscale=-1]
    \pgfmathsetmacro{\trianglew}{0.55}
    \pgfmathsetmacro{\triangleh}{1}

    \draw[draw=black] (-1.5,-1) -- (-1.5 -\trianglew,-1 -\triangleh) -- (-1.5 + \trianglew,-1 -\triangleh) -- cycle;
    \draw[draw=black] (0,-1) -- (0 -\trianglew,-1 -\triangleh) -- (0 + \trianglew,-1 -\triangleh) -- cycle;
    \draw[draw=black] (1.5,-1) -- (1.5 -\trianglew,-1 -\triangleh) -- (1.5 + \trianglew,-1 -\triangleh) -- cycle;

    \node at (0,0) (r) {};
    \node at (-1.5,-1) (a1) {};
    \node at (0,-1) (a2) {};
    \node at (1.5,-1)  (a3) {};

    \draw[] (r) -- (a1) node[midway, fill=white,draw=none,inner sep = 0.5] {$1$};
    \draw[] (r) -- (a2) node[midway, fill=white,draw=none,inner sep = 0.5] {$1$};
    \draw[] (r) -- (a3) node[midway, fill=white,draw=none,inner sep = 0.5] {$1$};

    \node[draw=white!0,fill=white!0]  at (-1.5,-1-0.6) (g) {$G$};
    \node[draw=white!0,fill=white!0]  at (0,-1-0.6) (g) {$G$};

    \node[draw=white!0,fill=white!0]  at (1.5,-1-0.6) (g) {$G$};

    \node[draw=white!0,fill=white!0]  at (0,-2.5) (g) {\color{white}$H_j$};
    \node[draw=white!0,fill=white!0]  at (0,-2.5) (g) {$H$};
  \end{tikzpicture}\hspace{0.5cm}
  \begin{tikzpicture}[scale=0.8,every node/.style={minimum size=0.2cm,inner sep=0,fill=gray!100,draw=black,circle},yscale=-1]
    \pgfmathsetmacro{\trianglew}{0.55}
    \pgfmathsetmacro{\triangleh}{1}

    \draw[draw=black] (-1.5,-1) -- (-1.5 -\trianglew,-1 -\triangleh) -- (-1.5 + \trianglew,-1 -\triangleh) -- cycle;
    \draw[draw=black] (0,-1) -- (0 -\trianglew,-1 -\triangleh) -- (0 + \trianglew,-1 -\triangleh) -- cycle;
    \draw[draw=black] (1.5,-1) -- (1.5 -\trianglew,-1 -\triangleh) -- (1.5 + \trianglew,-1 -\triangleh) -- cycle;

    \node at (0,0) (r) {};
    \node at (-1.5,-1) (a1) {};
    \node at (0,-1) (a2) {};
    \node at (1.5,-1)  (a3) {};

    \draw[] (r) -- (a1) node[midway, fill=white,draw=none,inner sep = 0.5] {$1$};
    \draw[] (r) -- (a2) node[midway, fill=white,draw=none,inner sep = 0.5] {$1$};
    \draw[] (r) -- (a3) node[midway, fill=white,draw=none,inner sep = 0.5] {$1$};

    \node[fill=none,draw=none]  at (-1.5,-1-0.7) (g) {\small $G_{u_{j,1}}$};
    \node[fill=none,draw=none]  at (0,-1-0.7) (g) {\small $G_{u_{j,2}}$};
    \node[fill=none,draw=none]  at (1.5,-1-0.7) (g) {\small $G_{u_{j,3}}$};

    \node[draw=white!0,fill=white!0]  at (0,-2.5) (g) {\color{white}$H_j$};
    \node[draw=white!0,fill=white!0]  at (0,-2.5) (g) {$H_j$};
  \end{tikzpicture} 
  \caption{The dummy set gadget $H$ and set gadget $H_j$.} \label{fig:set_gadget}
\end{figure}

\begin{figure}
  \centering
  \begin{tikzpicture}[scale=0.8,every node/.style={minimum size=0.2cm,inner sep=0,fill=gray!100,draw=black,circle},yscale=-1]
    \pgfmathsetmacro{\trianglew}{0.55}
    \pgfmathsetmacro{\triangleh}{1}

    \draw[draw=black] (-2.5,-1) -- (-2.5 -\trianglew,-1 -\triangleh) -- (-2.5 + \trianglew,-1 -\triangleh) -- cycle;
    \draw[draw=black] (-1,-1) -- (-1 -\trianglew,-1 -\triangleh) -- (-1 + \trianglew,-1 -\triangleh) -- cycle;
    \draw[draw=black] (1,-1) -- (1 -\trianglew,-1 -\triangleh) -- (1 + \trianglew,-1 -\triangleh) -- cycle;

    \node at (0,0.75) (pr) {};
    \node at (0,0) (r) {};
    \node at (-2.5,-1) (a1) {};
    \node at (-1,-1) (a2) {};
    \node at (1,-1)  (a3) {};

    \draw[] (pr) -- (r) node[midway, fill=white,draw=none,inner sep = 0.25] {$1$};
    \draw[] (r) -- (a1) node[midway, fill=white,draw=none,inner sep = 0.25] {$1$};
    \draw[] (r) -- (a2) node[midway, fill=white,draw=none,inner sep = 0.25] {$1$};
    \draw[] (r) -- (a3) node[midway, fill=white,draw=none,inner sep = 0.25] {$1$};

    \node[draw=white!0,fill=white!0]  at (-2.5,-1-0.7) (g) {$H_1$};
    \node[draw=white!0,fill=white!0]  at (-1,-1-0.7) (g) {$H_2$};
    \node[draw=white!0,fill=white!0]  at (1,-1-0.7) (g) {$H_n$};

    \node[draw=white!0,fill=white!0]  at (0,-1) (g) {$\dots{}$};

    \node[draw=white!0,fill=white!0]  at (1,0.25) (g) {\Large$T_1$};

    \draw[draw=white] (4.5,-1) -- (4.5 -\trianglew,-1 -\triangleh) -- (4.5 + \trianglew,-1 -\triangleh) -- cycle;
  \end{tikzpicture}\hspace{0.5cm}
  \begin{tikzpicture}[scale=0.8,every node/.style={minimum size=0.2cm,inner sep=0,fill=gray!100,draw=black,circle},yscale=-1]
    \pgfmathsetmacro{\trianglew}{0.55}
    \pgfmathsetmacro{\triangleh}{1}

    \draw[draw=black] (-2.5,-1) -- (-2.5 -\trianglew,-1 -\triangleh) -- (-2.5 + \trianglew,-1 -\triangleh) -- cycle;
    \draw[draw=black] (-1,-1) -- (-1 -\trianglew,-1 -\triangleh) -- (-1 + \trianglew,-1 -\triangleh) -- cycle;
    \draw[draw=black] (1,-1) -- (1 -\trianglew,-1 -\triangleh) -- (1 + \trianglew,-1 -\triangleh) -- cycle;

    \draw[draw=black] (2.5,-1) -- (2.5 -\trianglew,-1 -\triangleh) -- (2.5 + \trianglew,-1 -\triangleh) -- cycle;
    \draw[draw=black] (4.5,-1) -- (4.5 -\trianglew,-1 -\triangleh) -- (4.5 + \trianglew,-1 -\triangleh) -- cycle;

    \node at (0,0.75) (pr) {};
    \node at (0,0) (r) {};
    \node at (-2.5,-1) (a1) {};
    \node at (-1,-1) (a2) {};
    \node at (1,-1)  (a3) {};
    \node at (2.5,-1) (a4) {};
    \node at (4.5,-1)  (a5) {};

    \draw[] (pr) -- (r) node[midway, fill=white,draw=none,inner sep = 0.25] {$1$};
    \draw[] (r) -- (a1) node[midway, fill=white,draw=none,inner sep = 0.25] {$1$};
    \draw[] (r) -- (a2) node[midway, fill=white,draw=none,inner sep = 0.25] {$1$};
    \draw[] (r) -- (a3) node[midway, fill=white,draw=none,inner sep = 0.25] {$1$};
    \draw[] (r) -- (a4) node[midway, fill=white,draw=none,inner sep = 0.25] {$1$};
    \draw[] (r) -- (a5) node[midway, fill=white,draw=none,inner sep = 0.25] {$1$};

    \node[draw=none,fill=none]  at (-2.5,-1-0.7) (g) {$G_1$};
    \node[draw=none,fill=none]  at (-1,-1-0.7) (g) {$G_2$};
    \node[draw=none,fill=none]  at (1,-1-0.7) (g) {$G_m$};

    \node[draw=none,fill=none]  at (2.5,-1-0.7) (g) {$H$};
    \node[draw=none,fill=none]  at (4.5,-1-0.7) (g) {$H$};

    \draw [very thick,decorate,
        decoration = {calligraphic brace,mirror}] (4.5+0.7,-2.125) -- (2.5-0.7,-2.125);

    \node[fill=none,draw=none]  at (3.5,-2.5) (g) {$n-k$};

    \node[fill=none,draw=none]  at (0,-1) (dot1) {$\dots{}$};
    \node[fill=none,draw=none]  at (3.5,-1) (dot2) {$\dots{}$};

    \node[fill=none,draw=none]  at (1,0.4) (g) {\Large$T_2$};
  \end{tikzpicture}
  \caption{The two trees used in the reduction.} \label{fig:np_reduction}
\end{figure}
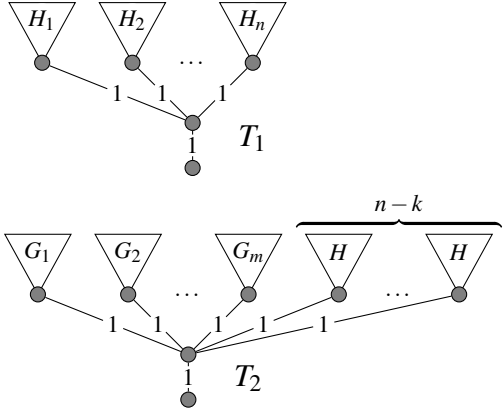

Based on the element gadgets, we next define set gadgets $H_i$ for each set $S_i \in S$ and a dummy set gadget $H$ (see Figure~\ref{fig:set_gadget}).
$H_i$ consists of a root node and the three element gadgets corresponding to the elements $u_{i,1}, u_{i,2}, u_{i,3}$ of $S_i$. The roots of the three element gadgets are connected to the new root via an edge of length $1$.
$H$ is constructed analogously, but uses three instances of $G$.

Now we construct the final trees $T_1$ and $T_2$.
In $T_1$, we encode each set $S_i \in S$ by first attaching the tree $H_i$ to the root via an edge of length $1$.
In $T_2$, we append for each element $u_i$ of the universe $U$ the tree $G_i$ and $n-k$ gadgets $H$ to the root, again via edges of length $1$.
For these two trees it then holds that $d(T_1,T_2) = 3n-2k$ if there is a covering of $U$ by $S$ and $d(T_1,T_2) > 3n-2k$ otherwise, which we prove in the following.

First, note that the total persistences of the two trees are $$|| T_1 || = 1 + n + 3n (2 + m(m+1)) = 1 + n + 6n + 3nm(m+1)$$ and 
\begin{align*}
  &|| T_2 || = \\ 
  &1 + 3k + 3k (1 + m(m+1)) + (n-k) + 3(n-k) (1 + m(m+1))\\ 
  &= 1 + n + 2k + 3n + 3nm(m+1).
\end{align*}
Thus, $$|| T_1 || - || T_2 || = 6n - (3n + 2k) = 3n - 2k.$$
\florian{poor formatting}

If there is a covering $C \subseteq S$ of $U$, then we can use the following sequence to convert $T_1$ into $T_2$.
For each of the $k$ sets in $C$, we contract the edge between the root of $T_1$ and the root of the corresponding gadget $H_i$.
Next, we transform each $(n-k)$ remaining gadgets $H_j$ into $H$ by contracting the three splitting edges in the element trees.
The total cost of this sequence are $k + 3(n-k) = 3n - 2k$.
Since these costs are exactly the difference in total persistence, the sequence is optimal.

In the case where no covering exists, we again know that $|| T_1 || - || T_2 || = 6n - (3n + 2k) = 3n - 2k$ and that $d(T_1,T_2) > 3n-2k$.
Thus, a sequence of cost exactly $3n-2k$ can only include deletions or shortening relabel operations.
Due to this property, such a sequence can only construct the $G$ and $G_i$ gadgets in $T_2$ from those in $T_2$.
Hence, it would imply a covering, leading to a contradiction.

\textit{(NP.)} The decision variant of our edit distance is clearly in \NP: given an edit sequence, we can check whether its cost is below the given threshold $c$. 
\end{proof}

\begin{figure}[]

\begin{subfigure}[b]{0.49\linewidth}
  \centering
  \resizebox{\linewidth}{!}{
  \begin{tikzpicture}[scale=0.4]
  \pgfmathsetmacro{\trianglew}{0.55*2}
  \pgfmathsetmacro{\triangleh}{1*2}

  \node[draw=none,fill=none,circle] at (-4, -1) (dummy) {};
  \node[draw=none,fill=none,circle] at (11, -1) (dummy) {};
  \node[draw=none,fill=none,circle] at (-4, 15) (dummy) {};
  \node[draw=none,fill=none,circle] at (11, 15) (dummy) {};
  
  \node[draw,circle,fill=gray!100] at (0, 0) (root_1) {};
  \node[draw,circle,fill=gray!100] at (0, 4) (s1_1) {};
  \node[draw,circle,fill=gray!100] at (-1, 5) (s2_1) {};
  \node[] at (-2, 12.33) (m1_1l) {$A$};
  \node[inner sep = 0, outer sep = 0, minimum size = 0] at (-2, 11) (m1_1) {};
  \draw[draw=black] (-2,11) -- (-2 -\trianglew,11 +\triangleh) -- (-2 + \trianglew,11 +\triangleh) -- cycle;
  \node[] at (2, 12.33) (m2_1l) {$B$};
  \node[inner sep = 0, outer sep = 0, minimum size = 0] at (2, 11) (m2_1) {};
  \draw[draw=black] (2,11) -- (2 -\trianglew,11 +\triangleh) -- (2 + \trianglew,11 +\triangleh) -- cycle;
  \node[] at (0, 10.33) (m3_1l) {$C$};
  \node[inner sep = 0, outer sep = 0, minimum size = 0] at (0, 9) (m3_1) {};
  \draw[draw=black] (0,9) -- (0 -\trianglew,9 +\triangleh) -- (0 + \trianglew,9 +\triangleh) -- cycle;

  \draw[gray,very thick] (root_1) -- (s1_1);
  \draw[gray,very thick] (s1_1) -- (m2_1);
  \draw[gray,very thick] (s1_1) -- (s2_1);
  \draw[gray,very thick] (s2_1) -- (m1_1);
  \draw[gray,very thick] (s2_1) -- (m3_1);
  
  \node[draw,circle,fill=gray!100] at (0+7, 0) (root_2) {};
  \node[draw,circle,fill=gray!100] at (0+7, 4) (s1_2) {};
  \node[draw,circle,fill=gray!100] at (1+7, 5) (s2_2) {};
  
  \node[] at (7+-2, 12.33) (m1_1l) {$A$};
  \node[inner sep = 0, outer sep = 0, minimum size = 0] at (7+-2, 11) (m1_2) {};
  \draw[draw=black] (7+-2,11) -- (7+-2 -\trianglew,11 +\triangleh) -- (7+-2 + \trianglew,11 +\triangleh) -- cycle;
  \node[] at (7+2, 12.33) (m2_1l) {$B$};
  \node[inner sep = 0, outer sep = 0, minimum size = 0] at (7+2, 11) (m2_2) {};
  \draw[draw=black] (7+2,11) -- (7+2 -\trianglew,11 +\triangleh) -- (7+2 + \trianglew,11 +\triangleh) -- cycle;
  \node[] at (7+0, 10.33) (m3_1l) {$C$};
  \node[inner sep = 0, outer sep = 0, minimum size = 0] at (7+0, 9) (m3_2) {};
  \draw[draw=black] (7+0,9) -- (7+0 -\trianglew,9 +\triangleh) -- (7+0 + \trianglew,9 +\triangleh) -- cycle;

  \draw[gray,very thick] (root_2) -- (s1_2);
  \draw[gray,very thick] (s1_2) -- (m1_2);
  \draw[gray,very thick] (s1_2) -- (s2_2);
  \draw[gray,very thick] (s2_2) -- (m2_2);
  \draw[gray,very thick] (s2_2) -- (m3_2);
  
  \end{tikzpicture}
  }
  \caption{horizontal}
  \label{fig:horizontal_instability}
\end{subfigure}
\begin{subfigure}[b]{0.49\linewidth}
  \centering
  \resizebox{0.95\linewidth}{!}{
  \begin{tikzpicture}[scale=0.4]
  \node[draw=none,fill=none,circle] at (-4, -1) (dummy) {};
  \node[draw=none,fill=none,circle] at (11, -1) (dummy) {};
  \node[draw=none,fill=none,circle] at (-4, 15) (dummy) {};
  \node[draw=none,fill=none,circle] at (11, 15) (dummy) {};
  
  \node[draw,circle,fill=gray!100] at (-2, 0) (root_1) {};
  \node[draw,circle,fill=gray!100] at (-2, 4) (s1_1) {};
  \node[draw,circle,fill=red!80, inner sep = 0, outer sep = 0, minimum width = 0.4cm] at (-2, 12) (m1_1) {x};
  \node[draw,circle,fill=red!80, inner sep = 0, outer sep = 0, minimum width = 0.4cm] at (2, 11.5) (m2_1) {y};

  \node[draw,circle,fill=gray!100] at (2, 6) (s2_1) {};
  \node[draw,circle,fill=gray!100] at (2, 7.5) (s3_1) {};
  \node[draw,circle,fill=gray!100] at (2, 9) (s4_1) {};
  \node[] at (0.5, 7) (s2'_1) {...};
  \node[] at (0.5, 8.5) (s3'_1) {...};
  \node[] at (0.5, 10) (s4'_1) {...};

  \draw[gray,very thick] (root_1) -- (s1_1);
  \draw[gray,very thick] (s1_1) -- (m1_1);
  \draw[gray,very thick] (s2_1) -- (s2'_1);
  \draw[gray,very thick] (s1_1) .. controls (2,4) .. (s2_1);
  \draw[gray,very thick] (s2_1) -- (s3_1);
  \draw[gray,very thick] (s3_1) -- (s4_1);
  \draw[gray,very thick] (s4_1) -- (m2_1);
  \draw[gray,very thick] (s3_1) -- (s3'_1);
  \draw[gray,very thick] (s4_1) -- (s4'_1);

  \node[draw,circle,fill=gray!100] at (-2+7, 0) (root_2) {};
  \node[draw,circle,fill=gray!100] at (-2+7, 4) (s1_2) {};
  \node[draw,circle,fill=red!80, inner sep = 0, outer sep = 0, minimum width = 0.4cm] at (2+7, 11.5) (m1_2) {x};
  \node[draw,circle,fill=red!80, inner sep = 0, outer sep = 0, minimum width = 0.4cm] at (-2+7, 12) (m2_2) {y};

  \node[draw,circle,fill=gray!100] at (-2+7, 6) (s2_2) {};
  \node[draw,circle,fill=gray!100] at (-2+7, 7.5) (s3_2) {};
  \node[draw,circle,fill=gray!100] at (-2+7, 9) (s4_2) {};
  \node[] at (-0.5+7, 7) (s2'_2) {...};
  \node[] at (-0.5+7, 8.5) (s3'_2) {...};
  \node[] at (-0.5+7, 10) (s4'_2) {...};

  \draw[gray,very thick] (root_2) -- (s1_2);
  \draw[gray,very thick] (s1_2) .. controls (2+7,4) .. (2+7,6);
  \draw[gray,very thick] (2+7,6) -- (m1_2);
  \draw[gray,very thick] (s1_2) -- (s2_2);
  \draw[gray,very thick] (s2_2) -- (s3_2);
  \draw[gray,very thick] (s3_2) -- (s4_2);
  \draw[gray,very thick] (s4_2) -- (m2_2);
  \draw[gray,very thick] (s2_2) -- (s2'_2);
  \draw[gray,very thick] (s3_2) -- (s3'_2);
  \draw[gray,very thick] (s4_2) -- (s4'_2);
  
  
  
  
  \end{tikzpicture}
  }
  \caption{vertical}
  \label{fig:vertical_instability}
\end{subfigure}
\caption{The two types of instabilities: In (a) the order of saddle points changes such that feature $C$ emerges from $B$ instead of $A$. Thus, we also call this a saddle swap. In (b) the two features $x$ and $y$ swap their position in the persistence-based ordering. Thus, mappings adhering to this hierarchy cannot map $y$ to $y$ and, more importantly, not map the subtrees on the path to $y$ to each other.}
\label{fig:instability}
\end{figure}
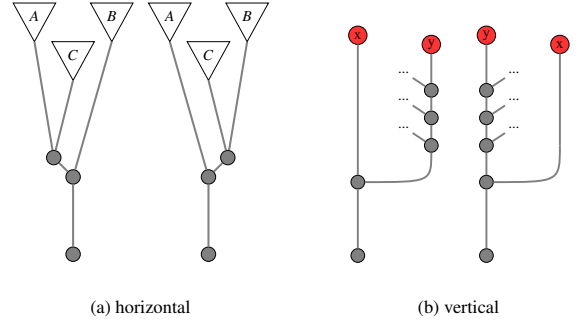

\section{Stability}
\label{sec:discussion}

In this section, we discuss how the unconstrained edit distance overcomes typical instabilities that other methods cannot handle.
We first recap the concepts of vertical and horizontal instability and identify the horizontal ones as the core problem.
We then discuss, on a conceptual level, why going from constrained to unconstrained edit distance precisely solves the problem of horizontal instability.
This gives an intuition on why our distance is not impeded by them.

\subsection*{Vertical and Horizontal Instability}

The quality of topological descriptors and distances on them is often measured by formal stability properties: do small perturbations in the scalar field only induce small distances between the abstractions?
For most descriptors, this property of course primarily depends on the chosen distance.

It turns out that many distances between merge trees suffer from instability.
Saikia et al.~\cite{DBLP:journals/cgf/SaikiaSW14} classified two very prominent types of instability in branch decompositions of merge trees: horizontal and vertical instability.
Indeed, vertical and horizontal instabilities have been identified as a main problem for edit distances between merge trees in several other works~\cite{DBLP:journals/cgf/SaikiaSW14,DBLP:journals/tvcg/SridharamurthyM20,wetzels2022branch}.

Figure~\ref{fig:instability} illustrates both types.
It can be seen that vertical instability is actually introduced by branch decompositions: it only appears when considering the persistence hierarchy of branches.
Accordingly, vertical instability does neither effect the constrained nor the unconstrained deformation-based edit distance.
This was also observed previously: vertical instabilities can be handled by branch decomposition-independent distances~\cite{wetzels2022branch,wetzels2022path,BeketayevYMWH14}, although inducing significant, but still polynomial computational overhead.

Horizontal instability appears not only in branch decompositions, but in merge trees in general.
It stems from the phenomenon described in Figure~\ref{fig:horizontal_instability}, which is often called a saddle swap.
In contrast to vertical instabilities, saddle-swaps remain a problem in basically all state-of-the-art approaches with efficient implementations.
Here, the constrained and unconstrained distances differ, which we will discuss next.

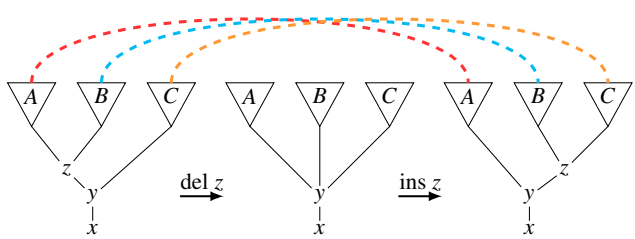
\begin{figure}
  \resizebox{\columnwidth}{!}{
  \begin{tikzpicture}[scale=0.6,every node/.style={minimum size=0.15cm,inner sep=1},yscale=-1]
    \pgfmathsetmacro{\trianglew}{0.55}
    \pgfmathsetmacro{\triangleh}{1}

    \begin{scope}[shift={(0,0)}]
    \node at (0,0.75) (a) {$x$};
    \node at (0,0) (b) {$y$};
    \node at (-0.6,-0.6) (c) {$z$};

    \node[inner sep = 0,outer sep = 0,minimum size = 0] at (-0.6-0.8,-0.6-1) (a1) {};
    \node[inner sep = 0,outer sep = 0,minimum size = 0] at (-0.6+0.8,-0.6-1) (a2) {};
    \node[inner sep = 0,outer sep = 0,minimum size = 0] at (-0.6+2.4,-0.6-1)  (a3) {};

    \draw[] (a) -- (b);
    \draw[] (b) -- (c);
    \draw[] (b) -- (a3);
    \draw[] (c) -- (a1);
    \draw[] (c) -- (a2);

    \draw[draw=black] (-1.4,-1.6) -- (-1.4 -\trianglew,-1.6 -\triangleh) -- (-1.4 + \trianglew,-1.6 -\triangleh) -- cycle;
    \node[draw=white!0,fill=white!0]  at (-1.4,-1.6-0.7) (g) {$A$};

    \draw[draw=black] (0.2,-1.6) -- (0.2 -\trianglew,-1.6 -\triangleh) -- (0.2 + \trianglew,-1.6 -\triangleh) -- cycle;
    \node[draw=white!0,fill=white!0]  at (0.2,-1.6-0.7) (g) {$B$};

    \draw[draw=black] (1.8,-1.6) -- (1.8 -\trianglew,-1.6 -\triangleh) -- (1.8 + \trianglew,-1.6 -\triangleh) -- cycle;
    \node[draw=white!0,fill=white!0]  at (1.8,-1.6-0.7) (g) {$C$};
    \end{scope}
    \begin{scope}[shift={(5,0)}]
      \node at (0.2,0.75) (a) {$x$};
      \node at (0.2,0) (b) {$y$};

      \node[inner sep = 0,outer sep = 0,minimum size = 0] at (-0.6-0.8,-0.6-1) (a1) {};
      \node[inner sep = 0,outer sep = 0,minimum size = 0] at (-0.6+0.8,-0.6-1) (a2) {};
      \node[inner sep = 0,outer sep = 0,minimum size = 0] at (-0.6+2.4,-0.6-1)  (a3) {};
  
      \draw[] (a) -- (b);
      \draw[] (b) -- (a1);
      \draw[] (b) -- (a2);
      \draw[] (b) -- (a3);
  
      \draw[draw=black] (-1.4,-1.6) -- (-1.4 -\trianglew,-1.6 -\triangleh) -- (-1.4 + \trianglew,-1.6 -\triangleh) -- cycle;
      \node[draw=white!0,fill=white!0]  at (-1.4,-1.6-0.7) (g) {$A$};
  
      \draw[draw=black] (0.2,-1.6) -- (0.2 -\trianglew,-1.6 -\triangleh) -- (0.2 + \trianglew,-1.6 -\triangleh) -- cycle;
      \node[draw=white!0,fill=white!0]  at (0.2,-1.6-0.7) (g) {$B$};
  
      \draw[draw=black] (1.8,-1.6) -- (1.8 -\trianglew,-1.6 -\triangleh) -- (1.8 + \trianglew,-1.6 -\triangleh) -- cycle;
      \node[draw=white!0,fill=white!0]  at (1.8,-1.6-0.7) (g) {$C$};
      \end{scope}
      \begin{scope}[shift={(10,0)}]
        \node at (0,0.75) (a) {$x$};
        \node at (0,0) (b) {$y$};
        \node at (0.8,-0.6) (c) {$z$};

        \node[inner sep = 0,outer sep = 0,minimum size = 0] at (-0.6-0.8,-0.6-1) (a1) {};
        \node[inner sep = 0,outer sep = 0,minimum size = 0] at (-0.6+0.8,-0.6-1) (a2) {};
        \node[inner sep = 0,outer sep = 0,minimum size = 0] at (-0.6+2.4,-0.6-1)  (a3) {};
    
        \draw[] (a) -- (b);
        \draw[] (b) -- (c);
        \draw[] (b) -- (a1);
        \draw[] (c) -- (a3);
        \draw[] (c) -- (a2);
    
        \draw[draw=black] (-1.4,-1.6) -- (-1.4 -\trianglew,-1.6 -\triangleh) -- (-1.4 + \trianglew,-1.6 -\triangleh) -- cycle;
        \node[draw=white!0,fill=white!0]  at (-1.4,-1.6-0.7) (g) {$A$};
    
        \draw[draw=black] (0.2,-1.6) -- (0.2 -\trianglew,-1.6 -\triangleh) -- (0.2 + \trianglew,-1.6 -\triangleh) -- cycle;
        \node[draw=white!0,fill=white!0]  at (0.2,-1.6-0.7) (g) {$B$};
    
        \draw[draw=black] (1.8,-1.6) -- (1.8 -\trianglew,-1.6 -\triangleh) -- (1.8 + \trianglew,-1.6 -\triangleh) -- cycle;
        \node[draw=white!0,fill=white!0]  at (1.8,-1.6-0.7) (g) {$C$};
        \end{scope}

        \draw[very thick,dashed,looseness=0.5,draw=red!80] (-1.4,-1.6 -\triangleh) to[out=270,in=270] (10+-1.4,-1.6 -\triangleh);
        \draw[very thick,dashed,looseness=0.5,draw=cyan!80] (0.2,-1.6 -\triangleh) to[out=270,in=270] (10+0.2,-1.6 -\triangleh);
        \draw[very thick,dashed,looseness=0.5,draw=orange!80] (1.8,-1.6 -\triangleh) to[out=270,in=270] (10+1.8,-1.6 -\triangleh);

        \draw[thick,-latex] (2.5-0.5,0) to (2.5+0.5,0);
        \node[] at (2.5,-0.33) {del $z$};

        \draw[thick,-latex] (7.5-0.5,0) to (7.5+0.5,0);
        \node[] at (7.5,-0.33) {ins $z$};

  \end{tikzpicture}
  }

  \vspace{0.5cm}
  
  \resizebox{\columnwidth}{!}{
  \begin{tikzpicture}[scale=0.6,every node/.style={minimum size=0.15cm,inner sep=1},yscale=-1]
    \pgfmathsetmacro{\trianglew}{0.55}
    \pgfmathsetmacro{\triangleh}{1}

    \begin{scope}[shift={(0,0)}]
    \node at (0,0.75) (a) {$x$};
    \node at (0,0) (b) {$y$};
    \node at (-0.6,-0.6) (c) {$z$};

    \node[inner sep = 0,outer sep = 0,minimum size = 0] at (-0.6-0.8,-0.6-1) (a1) {};
    \node[inner sep = 0,outer sep = 0,minimum size = 0] at (-0.6+0.8,-0.6-1) (a2) {};
    \node[inner sep = 0,outer sep = 0,minimum size = 0] at (-0.6+2.4,-0.6-1)  (a3) {};

    \draw[] (a) -- (b);
    \draw[] (b) -- (c);
    \draw[] (b) -- (a3);
    \draw[] (c) -- (a1);
    \draw[] (c) -- (a2);

    \draw[draw=black] (-1.4,-1.6) -- (-1.4 -\trianglew,-1.6 -\triangleh) -- (-1.4 + \trianglew,-1.6 -\triangleh) -- cycle;
    \node[draw=white!0,fill=white!0]  at (-1.4,-1.6-0.7) (g) {$A$};

    \draw[draw=black] (0.2,-1.6) -- (0.2 -\trianglew,-1.6 -\triangleh) -- (0.2 + \trianglew,-1.6 -\triangleh) -- cycle;
    \node[draw=white!0,fill=white!0]  at (0.2,-1.6-0.7) (g) {$B$};

    \draw[draw=black] (1.8,-1.6) -- (1.8 -\trianglew,-1.6 -\triangleh) -- (1.8 + \trianglew,-1.6 -\triangleh) -- cycle;
    \node[draw=white!0,fill=white!0]  at (1.8,-1.6-0.7) (g) {$C$};
    \end{scope}
    \begin{scope}[shift={(5,0)}]
      \node at (0.2,0.75) (a) {$x$};
      \node at (0.2,0) (b) {$y$};

      \node[inner sep = 0,outer sep = 0,minimum size = 0] at (-0.6,-0.2-1) (a1) {};
      \node[inner sep = 0,outer sep = 0,minimum size = 0] at (0.8,-0.2-1)  (a3) {};
  
      \draw[] (a) -- (b);
      \draw[] (b) -- (a1);
      \draw[] (b) -- (a3);
  
      \draw[draw=black] (-0.6,-1.2) -- (-0.6-\trianglew,-1.2 -\triangleh) -- (-0.6 + \trianglew,-1.2 -\triangleh) -- cycle;
      \node[draw=white!0,fill=white!0]  at (-0.6,-1.2-0.7) (g) {$A$};

      \draw[draw=black] (0.8,-1.2) -- (0.8 -\trianglew,-1.2 -\triangleh) -- (0.8 + \trianglew,-1.2 -\triangleh) -- cycle;
      \node[draw=white!0,fill=white!0]  at (0.8,-1.2-0.7) (g) {$C$};
      \end{scope}
      \begin{scope}[shift={(10,0)}]
        \node at (0,0.75) (a) {$x$};
        \node at (0,0) (b) {$y$};
        \node at (0.8,-0.6) (c) {$z$};

        \node[inner sep = 0,outer sep = 0,minimum size = 0] at (-0.6-0.8,-0.6-1) (a1) {};
        \node[inner sep = 0,outer sep = 0,minimum size = 0] at (-0.6+0.8,-0.6-1) (a2) {};
        \node[inner sep = 0,outer sep = 0,minimum size = 0] at (-0.6+2.4,-0.6-1)  (a3) {};
    
        \draw[] (a) -- (b);
        \draw[] (b) -- (c);
        \draw[] (b) -- (a1);
        \draw[] (c) -- (a3);
        \draw[] (c) -- (a2);
    
        \draw[draw=black] (-1.4,-1.6) -- (-1.4 -\trianglew,-1.6 -\triangleh) -- (-1.4 + \trianglew,-1.6 -\triangleh) -- cycle;
        \node[draw=white!0,fill=white!0]  at (-1.4,-1.6-0.7) (g) {$A$};
    
        \draw[draw=black] (0.2,-1.6) -- (0.2 -\trianglew,-1.6 -\triangleh) -- (0.2 + \trianglew,-1.6 -\triangleh) -- cycle;
        \node[draw=white!0,fill=white!0]  at (0.2,-1.6-0.7) (g) {$B$};
    
        \draw[draw=black] (1.8,-1.6) -- (1.8 -\trianglew,-1.6 -\triangleh) -- (1.8 + \trianglew,-1.6 -\triangleh) -- cycle;
        \node[draw=white!0,fill=white!0]  at (1.8,-1.6-0.7) (g) {$C$};
        \end{scope}

        \draw[very thick,dashed,looseness=0.5,draw=red!80] (-1.4,-1.6 -\triangleh) to[out=270,in=270] (10+-1.4,-1.6 -\triangleh);
        \draw[very thick,dashed,looseness=0.5,draw=orange!80] (1.8,-1.6 -\triangleh) to[out=270,in=270] (10+1.8,-1.6 -\triangleh);

        \draw[thick,-latex] (2.5-0.5,0) to (2.5+0.5,0);
        \node[] at (2.5,-0.33) {del $B, z$};

        \draw[thick,-latex] (7.5-0.5,0) to (7.5+0.5,0);
        \node[] at (7.5,-0.33) {ins $z, B$};
  \end{tikzpicture}
  }
  \caption{Saddle swaps in edit distances: deleting the node $z$ and then re-inserting it with different children makes it possible to move the feature $B$ horizontally from the subtree of feature $A$ to the subtree of feature $C$. If we do not allow the deletion of inner nodes, we need to delete $B$ completely and insert it again to arrive at the second tree. Thus, $B$ is not present in the mapping in the latter case.}
  \label{fig:saddle_swap}
\end{figure}

\subsection*{Saddle Swaps in Edit Distances}

In the case of tree edit distances, saddle swaps as described in Figure~\ref{fig:horizontal_instability} can actually be expressed as edit sequences.
Figure~\ref{fig:saddle_swap} illustrates how this can be done.
The cost of such a sequence then describes the difference between the two trees as we would expect intuitively, while keeping a semantically meaningful matching (not inhibited by constraints of the distances) between the features.
Figure~\ref{fig:saddle_swap} also shows why constrained edit sequences cannot achieve the same:
to perform saddle swaps in edit sequences, one needs to first delete an inner node or edge followed by re-inserting it.
In fact, there is no other way to move a subtree horizontally in the tree hierarchy.

Thus, saddle swaps exactly correspond to the type of changes only handled by edit distances which allow such operations on inner nodes.
If we allow these operations without any restrictions, then the distance becomes \NP-hard.
Whenever such operations on inner nodes are possible, it seems that we can produce a reduction in the vein of the one given in Section~\ref{sec:hardness}.

Nonetheless, the unconstrained deformation-based edit distance can therefore handle saddle swaps in a general way.
That this leads to very stable behavior.
In the next section, we validate this property experimentally, both on synthetic and real-world datasets.

\begin{figure}[]
  \centering
  \includegraphics[height=2.5cm]{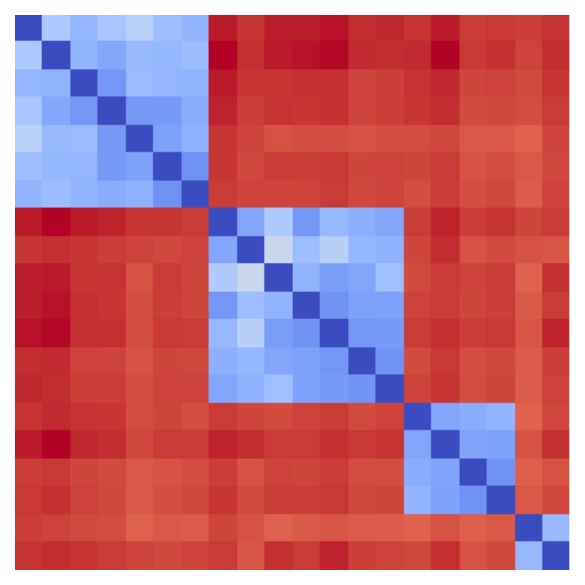}
  \hspace{0.001\linewidth}
  \includegraphics[height=2.5cm]{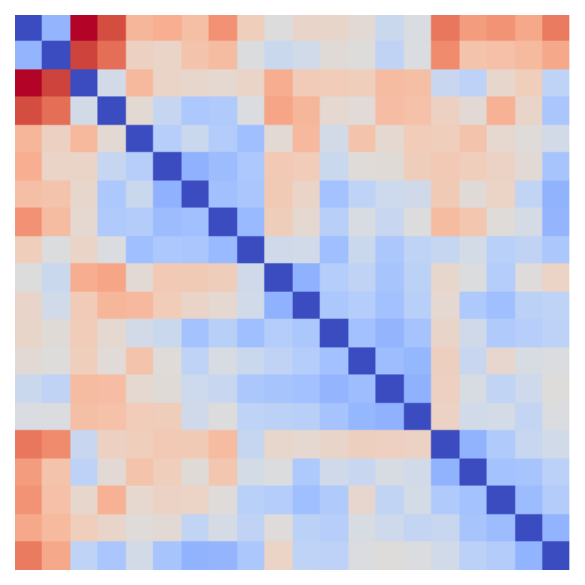}
  \hspace{0.001\linewidth}
  \includegraphics[height=2.5cm]{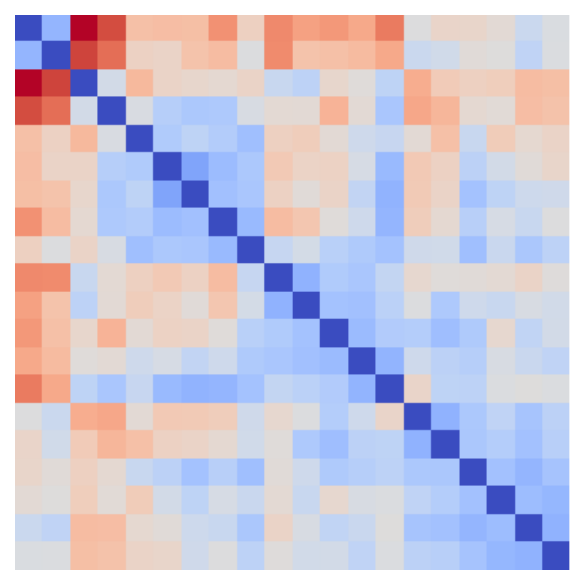}
  \hspace{0.001\linewidth}
  \includegraphics[height=2.5cm]{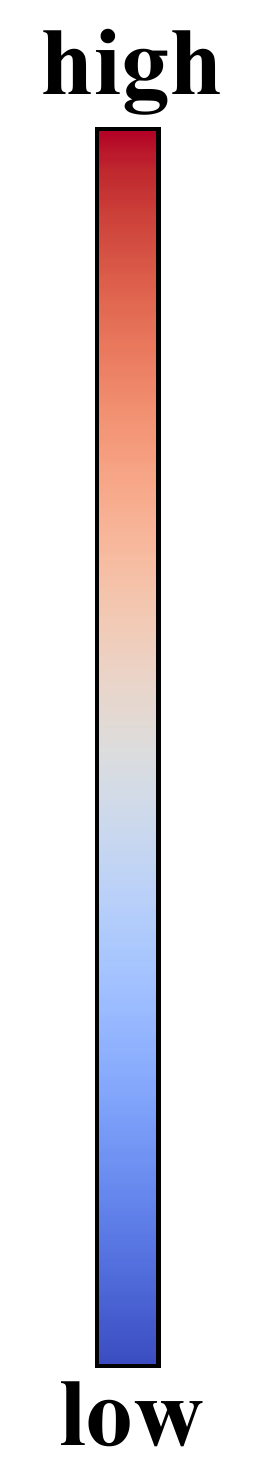}
  \caption{Distance matrices for the merge tree edit distances, the path mapping distance and the unconstrained deformation-based edit distance on the analytical example for vertical instabilities.}
  \label{fig:matrices_analytic_vertical}
\end{figure}

\begin{figure}[]
  \centering
  \includegraphics[height=2.5cm]{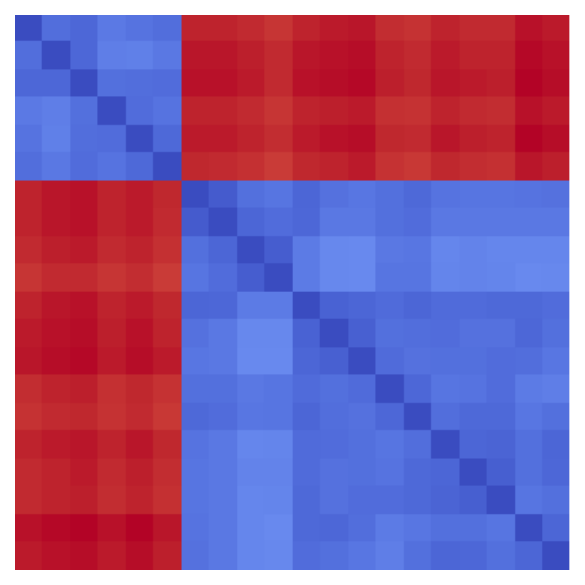}
  \hspace{0.001\linewidth}
  \includegraphics[height=2.5cm]{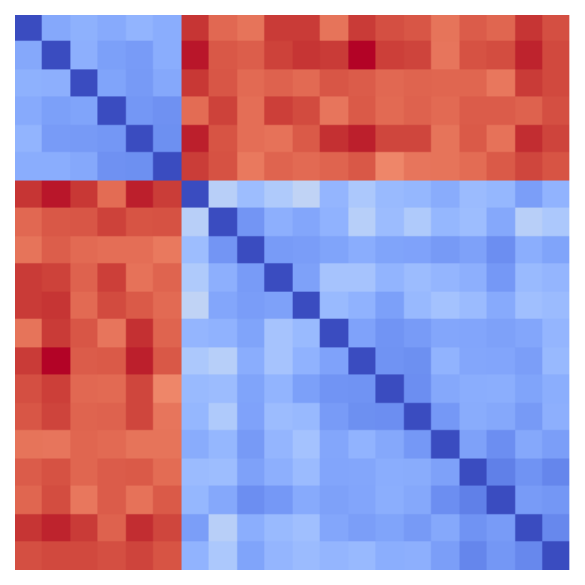}
  \hspace{0.001\linewidth}
  \includegraphics[height=2.5cm]{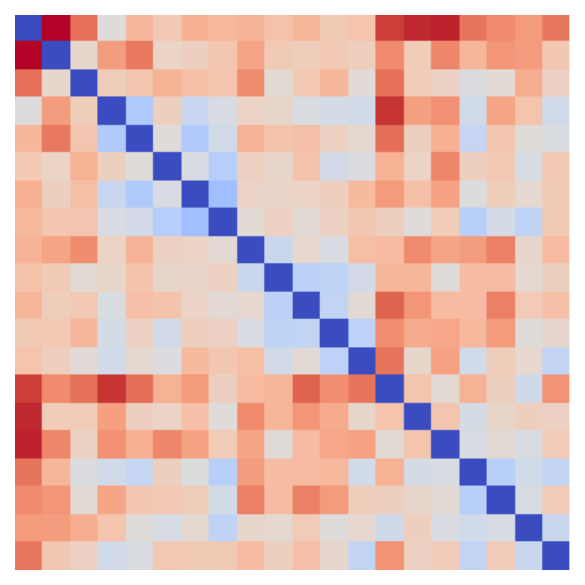}
  \hspace{0.001\linewidth}
  \includegraphics[height=2.5cm]{figures/cbar2.pdf}
  \caption{Distance matrices for the merge tree edit distances, the path mapping distance and the unconstrained deformation-based edit distance on the analytical example for horizontal instabilities.}
  \label{fig:matrices_analytic_horizontal}
\end{figure}

\begin{figure*}[]
  \centering
  \includegraphics[height=3.9cm]{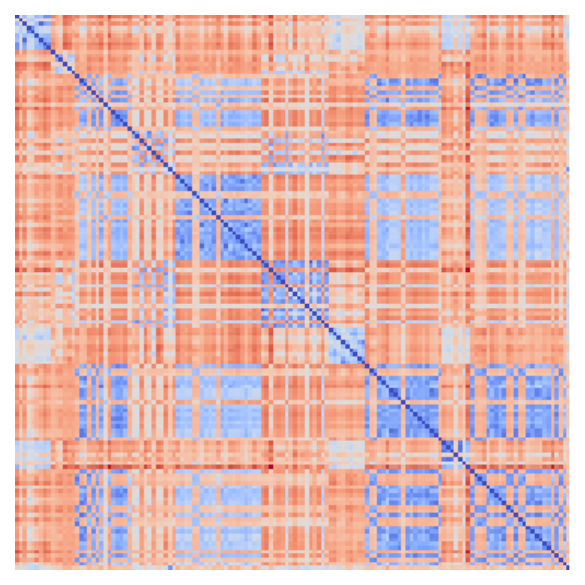}
  \hspace{0.01\linewidth}
  \includegraphics[height=3.9cm]{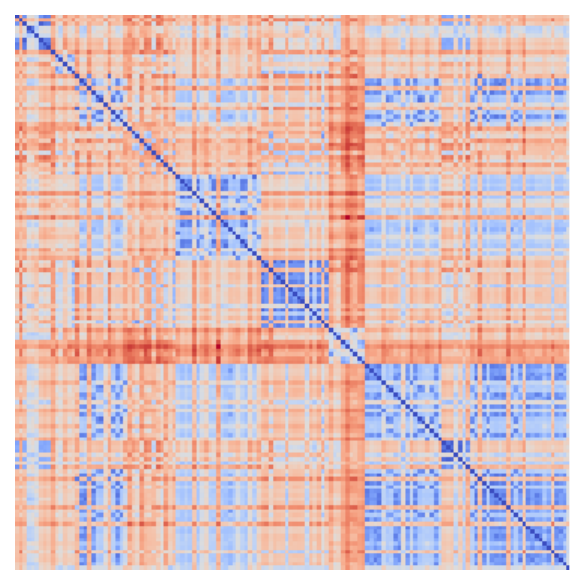}
  \hspace{0.01\linewidth}
  \includegraphics[height=3.9cm]{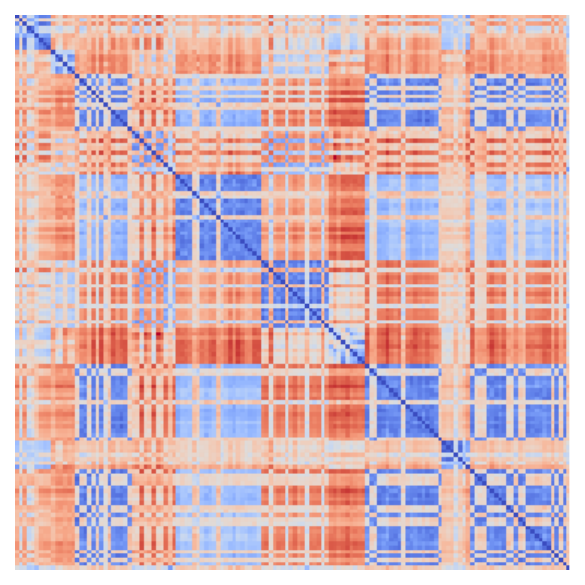}
  \hspace{0.01\linewidth}
  \includegraphics[height=3.9cm]{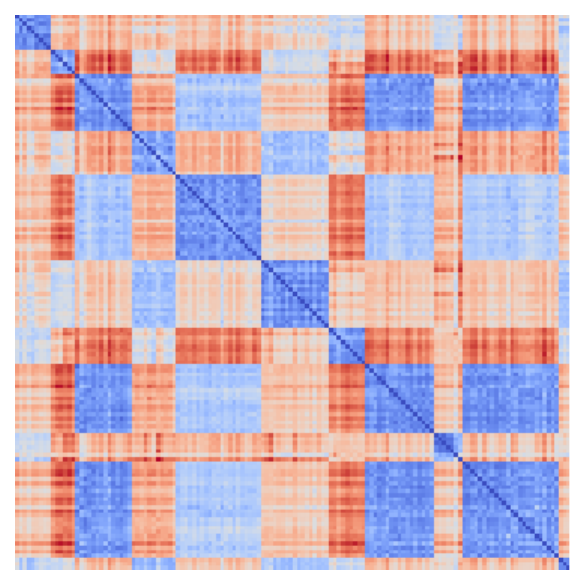}
  \hspace{0.01\linewidth}
  \includegraphics[height=3.9cm]{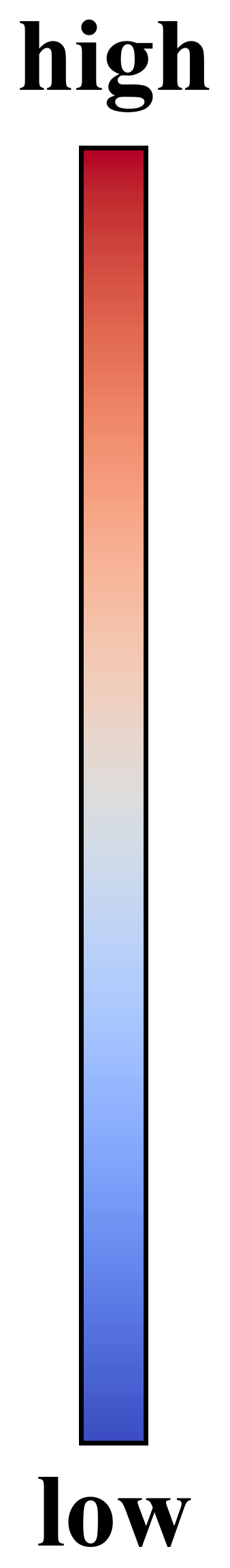}
  \caption{Four distance matrices of the TOSCA ensemble visualized as heatmaps. They were computed (from left to right) using the merge tree edit distance, the merge tree Wasserstein distance, the path mapping distance and the unconstrained deformation-based edit distance. No $\epsilon$-preprocessing was applied. From top left to bottom right, the clusters correspond with the following shapes: ``cat'', ``centaur'', ``david'', ``dog'', ``gorilla'', ``horse'', ``lioness'', ``michael'', ``seahorse''/``shark'', ``victoria'' and ``wolf''. Within all clusters the distances are generally small for the unconstrained distance (e.g.\ they range from 0.16 to 0.56 in the ``michael'' cluster), whereas they contain a lot of outliers/noise for the constrained distances (e.g.\ path mapping distances within ``michael'' range from 0.17 to 1.47). Also, the overall highest unconstrained and constrained deformation-based distances behave as expected, with the unconstrained one being significantly smaller (1.89 bs 2.34).}
  \label{fig:matrices_reduced}
\end{figure*}

\section{Experiments}
\label{sec:experiments}

In this section, we illustrate the stability of the unconstrained deformation-based edit distance experimentally.
We begin with experiments on two analytical datasets.
One is designed to specifically provoke vertical instabilities, while the other provokes horizontal instabilities.
On these datasets, we showcase that our distance measure can handle both types of instabilities in an isolated scenario.

Then, we consider the well-known TOSCA shape matching ensemble.
Indeed, these scalar fields contain both vertical and horizontal instabilities. 
An illustration of horizontal instabilities in this dataset can be found in Figure~\ref{fig:teaser}.
The unconstrained deformation-based edit distance significantly outperforms all other methods and yields almost perfect clusters, matching the various shapes of the ensemble onto each other.
This demonstrates that the distance can handle both types of instabilities at the same time, as well as that its advantages over existing approaches indeed matter in real-world (though non-scientific) datasets.

All topological preprocessing including simplification and merge tree computation was done in TTK~\cite{DBLP:journals/tvcg/TiernyFLGM18}, while the screenshots were rendered using in ParaView~\cite{paraviewBook} and VTK.

\subsection*{Analytical Example for Vertical Instability}
To test for vertical stability, we use the analytical example that was introduced in~\cite{wetzels2022branch} to illustrate the advantages of branch decomposition-independent distances for merge trees.
It is an ensemble consisting of 20 two-dimensional scalar fields, which exhibits strong vertical instability in its split trees.
All members consist of four main peaks, i.e.\ four main features in their merge tree.
One of them has five smaller maxima emerging from it.
Each member was constructed from the same baseline field with small perturbations applied to positions and heights of the peaks.
A stable distance should therefore also only express small noise-like distances.

However, this is only true for vertically stable distances.
The persistence hierarchy of the features differs between the members: sometimes the outlier peak is the highest maximum, sometimes it is not.
Edit mappings that have to adhere to this hierarchy therefore identify clusters, depending on where the outlier feature is placed in the branch decomposition.
In contrast, a perfect mapping matches the outlier features onto each other as well as the three others.
It therefore only conveys the original noise.

Figure~\ref{fig:matrices_analytic_vertical} shows the distance matrices of three different distances on this dataset.
Here, we first computed all distances and then applied a clustermap to highlight clustered results.
It can be observed that the merge tree edit distance yields a clustered distance matrix while the two deformation-based distances show the ``correct'' noise-like distances.
These observations concur with the theoretic properties of  the distances.
While the merge tree edit distance is branch decomposition-based and shows false clusters, the path mapping distance was designed to overcome exactly this problem.
As expected, the unconstrained deformation based behaves almost identical to the path mapping distance, since perfect matchings on this dataset do not require any saddle swaps.

\subsection*{Analytical Example for Horizontal Instability}

Our second experiment is performed on another analytical ensemble, which we designed specifically for this paper to test for horizontal stability.
The overall idea is similar to the first experiment.
The 20 two-dimensional scalar fields consist of three distinct features (again they differ in the amount of emerging side peaks), and due to small perturbations the nesting of these varies between the members.
The smallest of the main peaks first branches from one of the two other peaks.
The order differs within the members of the ensemble.
This leads to saddle swaps being necessary to match the three features correctly.
Since we constructed this dataset with a strict ordering in height of the features, it does not contain any vertical instability.

The test for stability is then analogous to the first example: if the distance exhibits clusters, it is horizontally unstable.
Again, the results presented in Figure~\ref{fig:matrices_analytic_horizontal} match the expected outcome.
The merge tree edit distance is unstable, as is the path mapping distance, unlike on the first dataset (which demonstrates that it only solves the problem of vertical instability).
The unconstrained deformation-based distance again represents the noise in the data correctly.

\begin{figure}[h]
  \centering
  \includegraphics[width=0.85\linewidth]{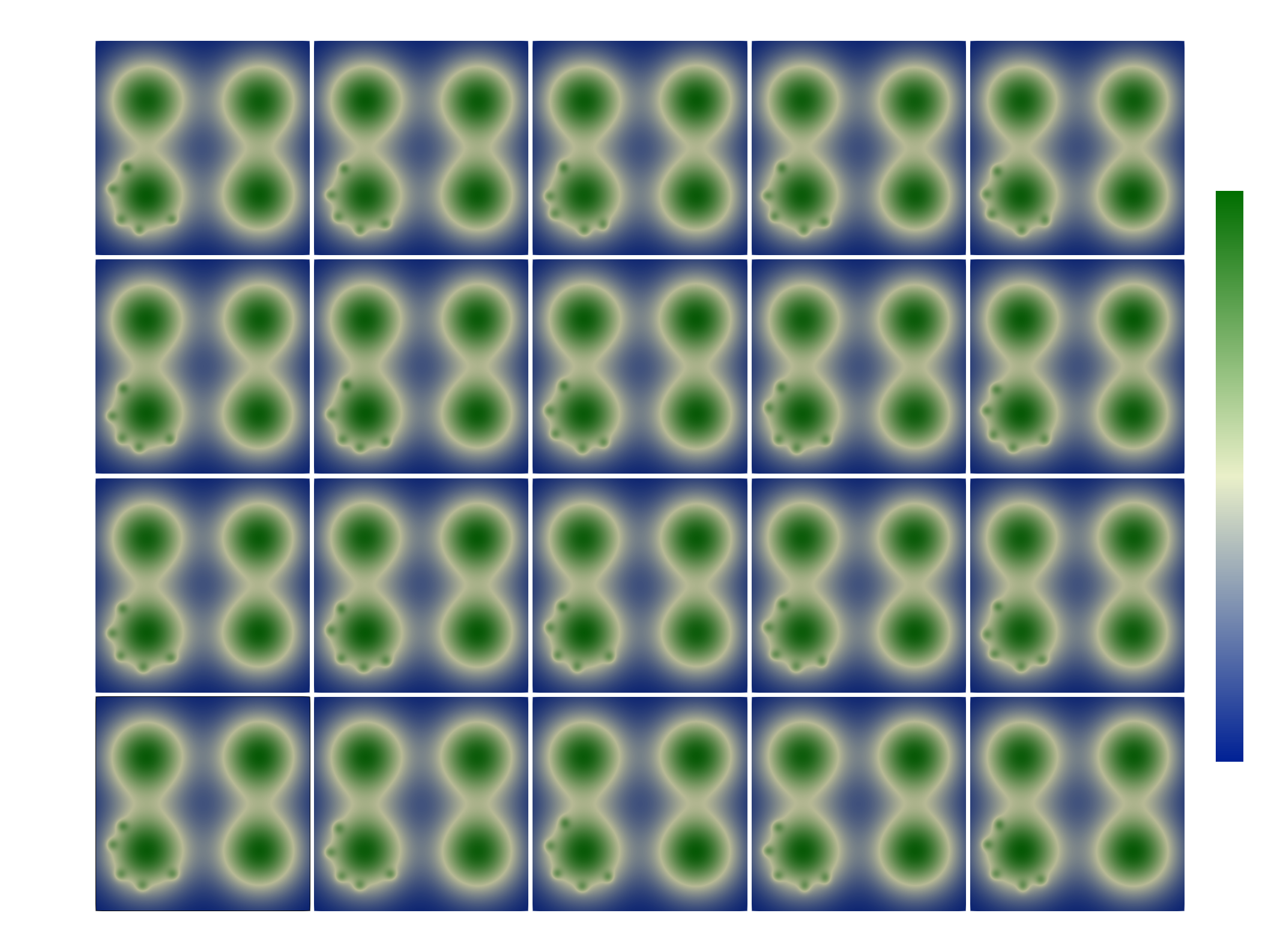}
  \includegraphics[width=0.85\linewidth]{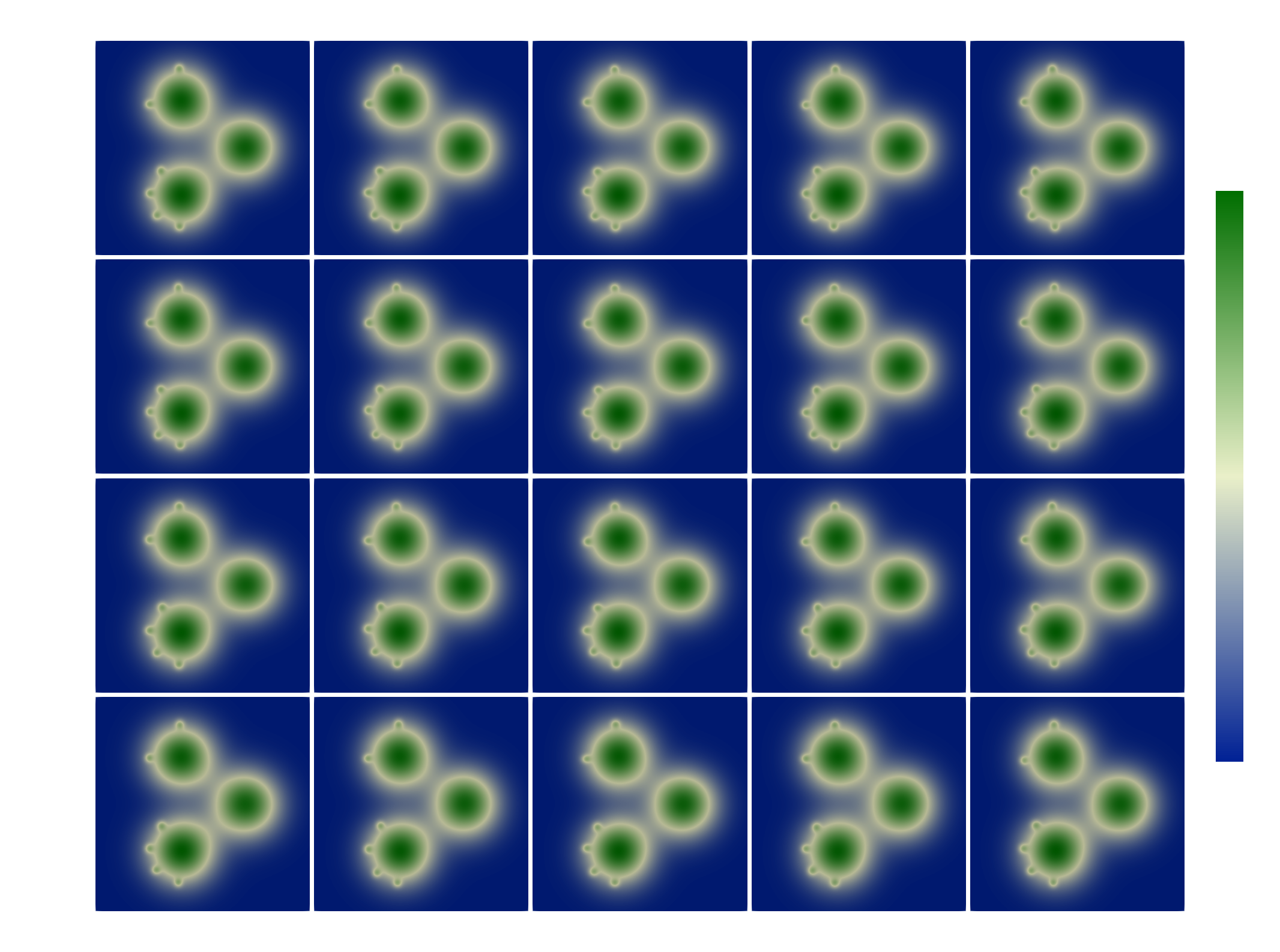}
  \caption{The two analytical ensembles intuitively only show noise-like differences. However, the first one exhibits vertical instabilities: which of the four peaks is the most persistent one differs between the members. The second one exhibits horizontal instabilities: the members differ in which of the two peaks on the left first merges with the peak on the right.}
  \label{fig:ensembles_analytical}
\end{figure}

\subsection*{TOSCA}

Our main experiment was performed on the TOSCA non-rigid world dataset~\cite{DBLP:series/mcs/BronsteinBK09}.
It is a shape matching ensemble consisting of several human and animal shapes in varying poses.
The surface meshes have a scalar field attached representing the average geodesic distance of the vertices~\cite{HilagaSKK01}.
We use the same field that was also used by Sridharamurthy et al. in~\cite{DBLP:journals/tvcg/SridharamurthyM20}.
Figure~\ref{fig:shapes_tosca} shows four example members from four different shapes.
A good distance should identify the different shapes correctly as clusters.
Furthermore, similar shapes (e.g.\ two different human shapes, like ``victoria'' and ``david'' in Figure~\ref{fig:shapes_tosca}) should also yield small distances.

We applied topological simplification (using TTK~\cite{DBLP:journals/tvcg/TiernyFLGM18}) with a relative threshold of 6\% to remove noise and reduce the size of the trees.
Then, we computed the split trees with TTK.
As a last preprocessing step, we removed all merge trees with more than 26 vertices to get feasible computation times.
This mainly removed members from the ``lioness'' shape, which seems to be an outlier cluster in terms of size or noise of the split trees.

On this reduced ensemble, we computed the full distance matrix using four different edit distances: the merge tree edit distance by Sridharamurthy et al.~\cite{DBLP:journals/tvcg/SridharamurthyM20}, the merge tree Wasserstein distance by Pont et al.~\cite{DBLP:journals/tvcg/PontVDT22}, the path mapping distance by Wetzels et al.~\cite{wetzels2022path} and the new unconstrained deformation-based edit distance.
The three existing distances are all implemented in TTK, which we used to derive the matrices.
The results can be seen in Figure~\ref{fig:matrices_reduced}.
Clearly, the unconstrained distance yields the cleanest results with perfectly identified clusters corresponding to the various human or animal shapes (the entries are ordered by shape first).
All other distances heavily suffer from noise in the distances, which overshadows the clusters present in the data.
Furthermore, it can be observed that the clusters are also visually more distinct when using the path mapping distance in comparison to the two branch decomposition-based approaches.
This implies that the ensemble not only exhibits horizontal instabilities, but also vertical ones, which are both handled well by the unconstrained edit distance.

\begin{figure}[h]
  \begin{subfigure}{\linewidth}
  \centering
  \includegraphics[height=2.5cm]{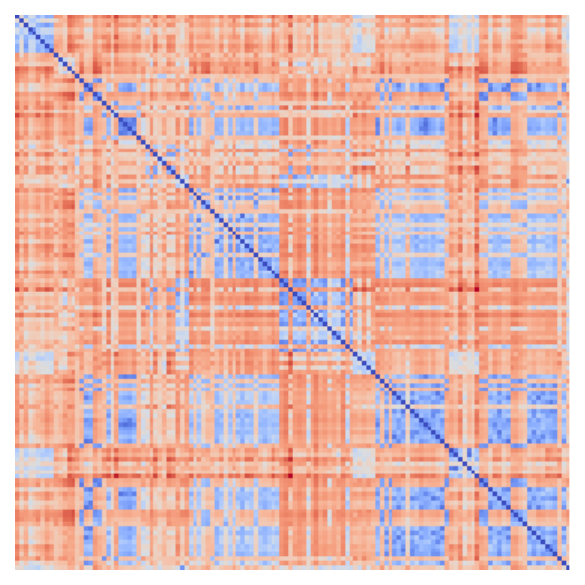}
  \hspace{0.001\linewidth}
  \includegraphics[height=2.5cm]{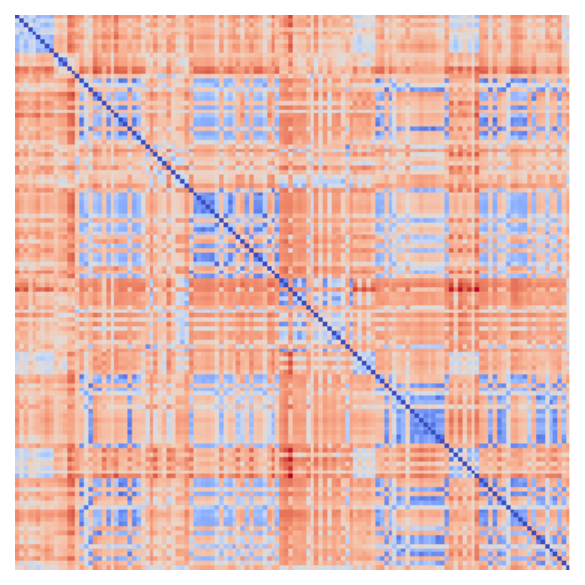}
  \hspace{0.001\linewidth}
  \includegraphics[height=2.5cm]{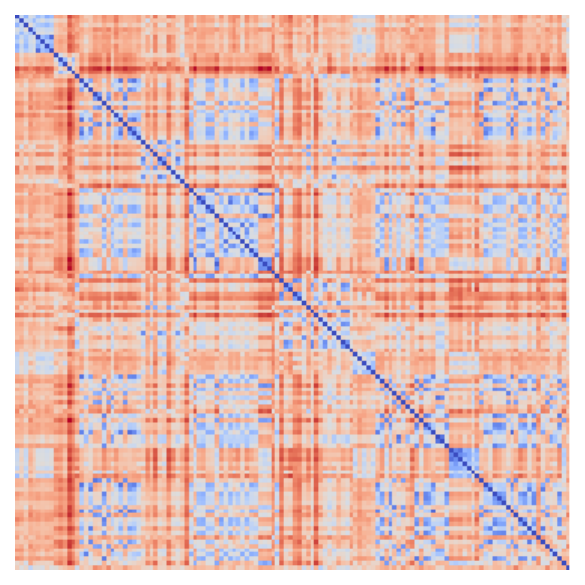}
  \hspace{0.001\linewidth}
  \includegraphics[height=2.5cm]{figures/cbar2.pdf}
  \caption{Distance matrices for the merge tree edit distance with preprocessing.}
  \label{fig:matrices_reduced_eps_mted}
\end{subfigure}
\begin{subfigure}{\linewidth}
  \centering
  \includegraphics[height=2.5cm]{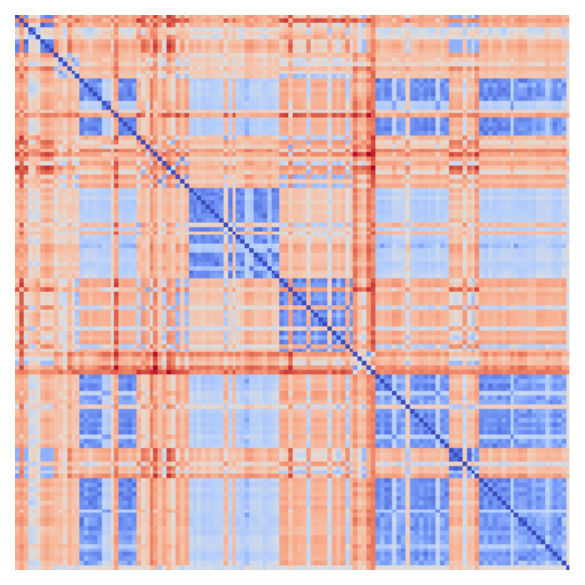}
  \hspace{0.001\linewidth}
  \includegraphics[height=2.5cm]{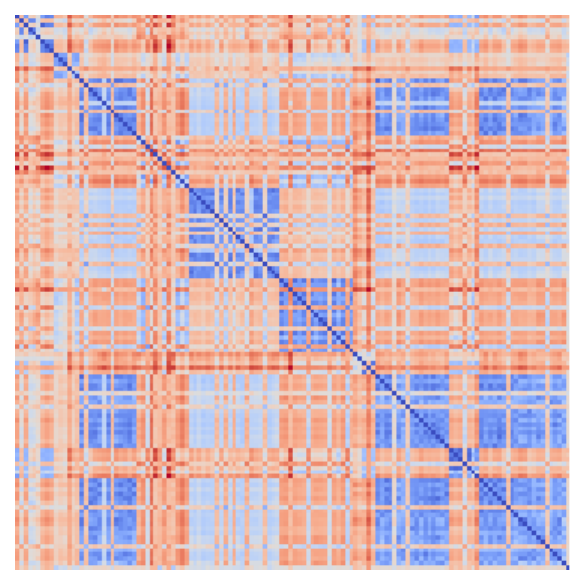}
  \hspace{0.001\linewidth}
  \includegraphics[height=2.5cm]{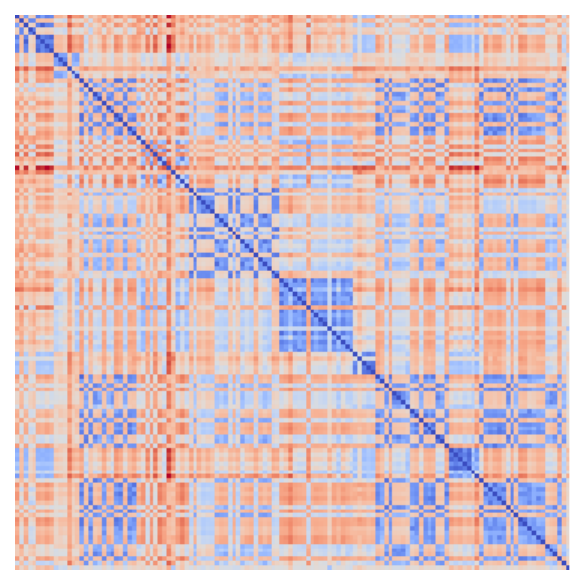}
  \hspace{0.001\linewidth}
  \includegraphics[height=2.5cm]{figures/cbar2.pdf}
  \caption{Distance matrices for the merge tree Wasserstein distance with preprocessing.}
  \label{fig:matrices_reduced_eps_ws}
\end{subfigure}
\caption{Distance matrices for the merge tree edit distance (a) and the merge tree Wasserstein distance (b) with $\epsilon$-preprocessing. The $\epsilon$ values were (left to right) 2\%, 5\% and 10\% of the scalar range.}
\label{fig:matrices_reduced_eps}
\end{figure}

\subsection*{Effect of Preprocessing}
To compare the new distance with the similar method of $\epsilon$-preprocessing, we also computed the distance matrices of the merge tree edit distance and the merge tree Wasserstein distance for different values of $\epsilon$.
As argued in Section~\ref{sec:discussion}, both approaches try to fix the same problem, but unconstrained edit distances do so in a more general way.

Our experiments show that applying the preprocessing with small threshold values (under 10\% of the scalar range) still leads to cluttered distance matrices and noisy clusters.
The different matrices can be seen in Figure~\ref{fig:matrices_reduced_eps}.
Another interesting observation is that for both distances, the results get worse with a threshold of 10\% in comparison to 5\%.
Thus, we conclude that the preprocessing exhibits a less predictable behavior than the more generic solution of using unconstrained distances.

\begin{figure}[]
  \centering
  \includegraphics[width=0.75\linewidth]{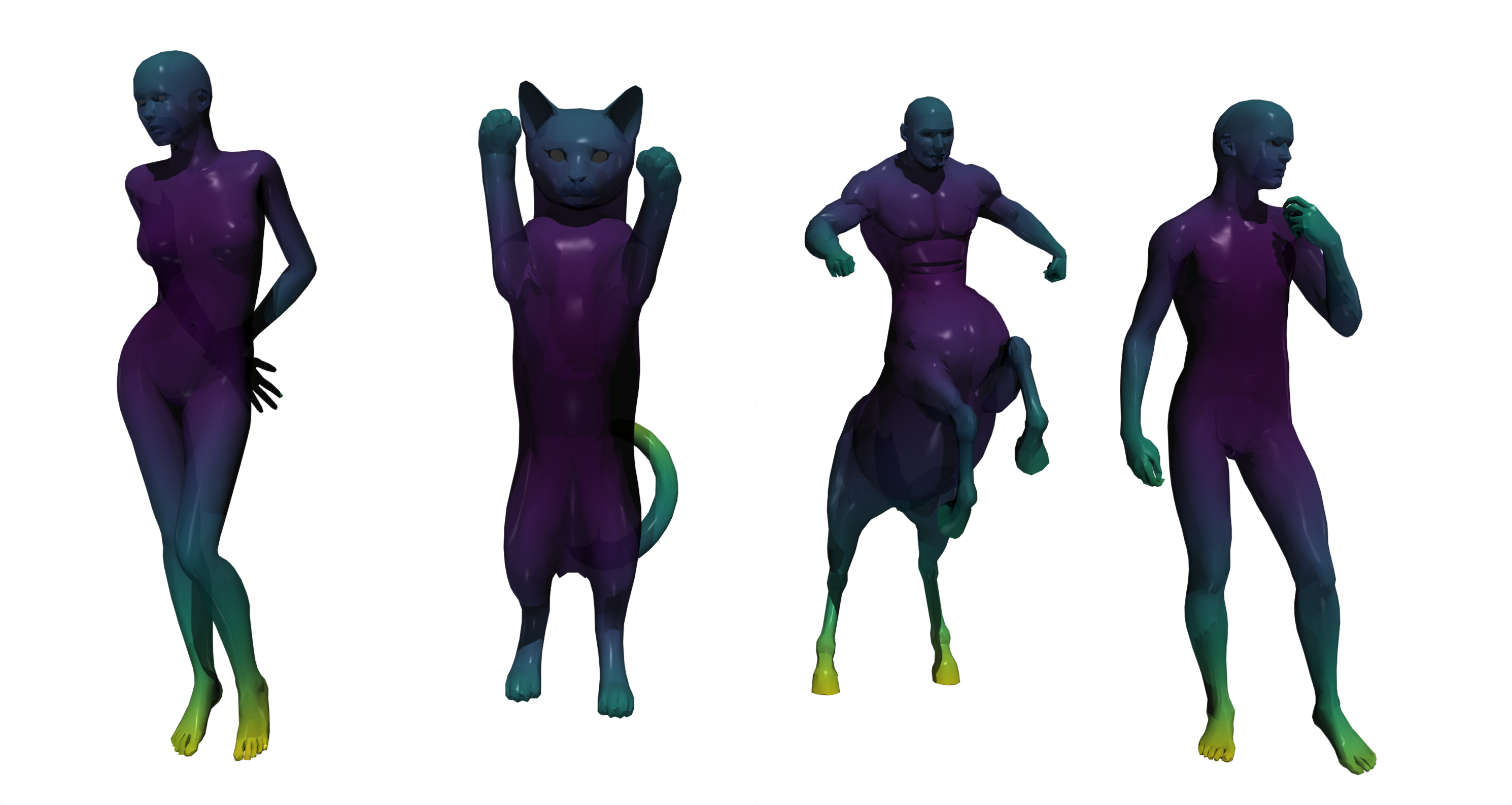}
  \caption{Four different shapes from the TOSCA ensemble: ``victoria'', ``cat'', ``centaur'' and ``david'' (left to right).}
  \label{fig:shapes_tosca}
\end{figure}

\subsection*{Runtime Performance}
We now discuss the observed runtime of the IP formulation of the unconstrained deformation-based edit distance.
Since IP is known to be \NP-complete, exponential running times are to be expected.
We implemented the reduction to IP in Python and used the Python binding of \gurobi{} through \pulp{}.

Up to trees of about 20 vertices, runtimes for single distance computations are in the range of seconds. Beyond this, they become infeasible rather quickly.
In the setting of computing the full distance matrix on the TOSCA ensemble, our implementation was able to handle distances for trees of up to 26 vertices.
After removing all trees with more than 26 vertices, 138 members remained.
It took about 5 days using 100 threads to compute all 9453 distances for the matrix.
In comparison, for all other distances and all values of $\epsilon$, computing the matrix took only several seconds, using the C++ implementation in TTK with ParaView.

\section{Conclusion and Future Work}
\label{sec:conclusion}

In this paper, we studied the unconstrained deformation-based edit distance for merge trees.
Compared to other edit distances for merge trees, this distance exhibits significantly improved robustness against both vertical and horizontal instabilities.
We described the conceptual background of these advantages and also demonstrated the stability experimentally: handling saddle swaps is a very effective way to increase the expressiveness of merge tree edit distances.
Although we established its \NP-completeness, we were able to derive an integer programming implementation to showcase the distance on reasonably sized (meaning small enough for our approach to terminate in reasonable time) instances.

Our results validate the quality of the deformation-based edit distance which was proposed by Wetzels et al.~\cite{wetzels2022path} previously, but only implemented in a constrained variant.
In future work, we want to study theoretical stability properties to complement the presented experimental insights.
Furthermore, our findings give rise to several strategies for designing improved \emph{tractable}
distances in the future.

\subsection*{Possible Strategies for Future Methods}

In our experiments, unconstrained edit distances expressed high potential in terms of comparison quality and stability.
They clearly outperform simpler constrained edit distances, independent of the preprocessing threshold used.
Nonetheless, $\epsilon$-preprocessing has proven to be effective on practical datasets as well~\cite{DBLP:journals/tvcg/SridharamurthyM20,DBLP:journals/tvcg/PontVDT22} and has significant advantages in terms of running time.
Thus, we believe that future methods for comparing merge trees through edit distances should focus on utilizing/allowing some form of actual saddle swaps (e.g.\ allow for a small or constant number of inner deletions), as well as combining such approaches with $\epsilon$-preprocessing with small threshold values.

In addition, further optimization of the proposed integer linear program should be considered.
There could be a lot of potential in improved reencoding or other strategies to reduce the size of the linear program.
An interesting question to study could be from which time bound on the more expensive encoding techniques should be applied.
Furthermore, optimizations similar to those by Hong et al.~\cite{DBLP:conf/cocoa/HongKY17} for classic tree edit distances could be considered.

During our experiments, we observed a significant impact of the IP solver used, depending on the applied reencoding strategies.
Thus, apart from further optimizing the presented strategies, properly studying the effects of different solvers should be considered.

Overall, the distance studied in this paper has potential of being applied directly on small instances and leading to very effective heuristic distances with more practical running times.

\subsection*{Supplementary Material}
This manuscript is accompanied by supplementary material:
We provide the source code of our Python implementation as well as the synthetic dataset for horizontal stability.
We will release the code in an open source repository~\cite{repository}.

\iflorian{ToDos:}
\iflorian{Figure placement}
\iflorian{update and check matrices}
\iflorian{report running times}
\iflorian{code supplement}
\iflorian{Maybe add subsection for distance matrix with time limit}

\acknowledgments{
The authors wish to thank Raghavendra Sridharamurthy for providing the pre-processed TOSCA dataset and Marvin Petersen for helpful discussions.
This work is funded by the Deutsche Forschungsgemeinschaft (DFG, German Research Foundation) – 442077441 - as well as by the European Research Council (ERC) under the European Union's Horizon 2020 research and innovation programme (EngageS: grant agreement No.~{820148}).
}

\bibliographystyle{abbrv-doi}

\bibliography{paper}

\begin{thebibliography}{10}

\bibitem{paraviewBook}
J.~P. Ahrens, B.~Geveci, and C.~C. Law.
\newblock Paraview: An end-user tool for large-data visualization.
\newblock In C.~D. Hansen and C.~R. Johnson, eds., {\em The Visualization
  Handbook}, pp. 717--731. Academic Press / Elsevier, 2005. doi: {{%
10\hspace{.1pt}\discretionary{.}{%
}{.}\hspace{.4pt}1016\discretionary{/}{%
}{/}b978\discretionary{%
}{-}{-}012387582\discretionary{%
}{-}{-}2\discretionary{/}{%
}{/}50038\discretionary{%
}{-}{-}1}}


\bibitem{DBLP:conf/3dor/BauerFL16}
U.~Bauer, B.~D. Fabio, and C.~Landi.
\newblock An edit distance for reeb graphs.
\newblock In A.~Ferreira, A.~Giachetti, and D.~Giorgi, eds., {\em 9th
  Eurographics Workshop on 3D Object Retrieval, 3DOR@Eurographics 2016, Lisbon,
  Portugal, May 8, 2016}. Eurographics Association, 2016. doi: {{%
10\hspace{.1pt}\discretionary{.}{%
}{.}\hspace{.4pt}2312\discretionary{/}{%
}{/}3dor\hspace{.1pt}\discretionary{.}{%
}{.}\hspace{.4pt}20161084}}


\bibitem{BeketayevYMWH14}
K.~Beketayev, D.~Yeliussizov, D.~Morozov, G.~H. Weber, and B.~Hamann.
\newblock Measuring the distance between merge trees.
\newblock In P.~Bremer, I.~Hotz, V.~Pascucci, and R.~Peikert, eds., {\em
  Topological Methods in Data Analysis and Visualization III, Theory,
  Algorithms, and Applications}, pp. 151--165. Springer, 2014. doi: {{%
10\hspace{.1pt}\discretionary{.}{%
}{.}\hspace{.4pt}1007\discretionary{/}{%
}{/}978\discretionary{%
}{-}{-}3\discretionary{%
}{-}{-}319\discretionary{%
}{-}{-}04099\discretionary{%
}{-}{-}8\_10}}


\bibitem{BestuzhevaEtal2021OO}
K.~Bestuzheva, M.~Besan{\c{c}}on, W.-K. Chen, A.~Chmiela, T.~Donkiewicz, J.~van
  Doornmalen, L.~Eifler, O.~Gaul, G.~Gamrath, A.~Gleixner, L.~Gottwald,
  C.~Graczyk, K.~Halbig, A.~Hoen, C.~Hojny, R.~van~der Hulst, T.~Koch,
  M.~L{\"u}bbecke, S.~J. Maher, F.~Matter, E.~M{\"u}hmer, B.~M{\"u}ller, M.~E.
  Pfetsch, D.~Rehfeldt, S.~Schlein, F.~Schl{\"o}sser, F.~Serrano, Y.~Shinano,
  B.~Sofranac, M.~Turner, S.~Vigerske, F.~Wegscheider, P.~Wellner, D.~Weninger,
  and J.~Witzig.
\newblock {The SCIP Optimization Suite 8.0}.
\newblock Technical report, Optimization Online, December 2021.

\bibitem{treeEditSurvey}
P.~Bille.
\newblock A survey on tree edit distance and related problems.
\newblock {\em Theoretical Computer Science}, 337(1-3):217--239, 2005. doi: {{%
10\hspace{.1pt}\discretionary{.}{%
}{.}\hspace{.4pt}1016\discretionary{/}{%
}{/}j\hspace{.1pt}\discretionary{.}{%
}{.}\hspace{.4pt}tcs\hspace{.1pt}\discretionary{.}{%
}{.}\hspace{.4pt}2004\hspace{.1pt}\discretionary{.}{%
}{.}\hspace{.4pt}12\hspace{.1pt}\discretionary{.}{%
}{.}\hspace{.4pt}030}}


\bibitem{DBLP:journals/tvcg/BollenTL23}
B.~Bollen, P.~Tennakoon, and J.~A. Levine.
\newblock Computing a stable distance on merge trees.
\newblock {\em {IEEE} Trans. Vis. Comput. Graph.}, 29(1):1168--1177, 2023. doi:
  {{%
10\hspace{.1pt}\discretionary{.}{%
}{.}\hspace{.4pt}1109\discretionary{/}{%
}{/}TVCG\hspace{.1pt}\discretionary{.}{%
}{.}\hspace{.4pt}2022\hspace{.1pt}\discretionary{.}{%
}{.}\hspace{.4pt}3209395}}


\bibitem{DBLP:journals/talg/BringmannGMW20}
K.~Bringmann, P.~Gawrychowski, S.~Mozes, and O.~Weimann.
\newblock Tree edit distance cannot be computed in strongly subcubic time
  (unless {APSP} can).
\newblock {\em {ACM} Trans. Algorithms}, 16(4):48:1--48:22, 2020. doi: {{%
10\hspace{.1pt}\discretionary{.}{%
}{.}\hspace{.4pt}1145\discretionary{/}{%
}{/}3381878}}


\bibitem{DBLP:series/mcs/BronsteinBK09}
A.~M. Bronstein, M.~M. Bronstein, and R.~Kimmel.
\newblock {\em Numerical Geometry of Non-Rigid Shapes}.
\newblock Monographs in Computer Science. Springer, 2009. doi: {{%
10\hspace{.1pt}\discretionary{.}{%
}{.}\hspace{.4pt}1007\discretionary{/}{%
}{/}978\discretionary{%
}{-}{-}0\discretionary{%
}{-}{-}387\discretionary{%
}{-}{-}73301\discretionary{%
}{-}{-}2}}


\bibitem{interleaving_distance}
F.~Chazal, D.~Cohen{-}Steiner, M.~Glisse, L.~J. Guibas, and S.~Oudot.
\newblock Proximity of persistence modules and their diagrams.
\newblock In J.~Hershberger and E.~Fogel, eds., {\em Proceedings of the 25th
  {ACM} Symposium on Computational Geometry, Aarhus, Denmark, June 8-10, 2009},
  pp. 237--246. {ACM}, 2009. doi: {{%
10\hspace{.1pt}\discretionary{.}{%
}{.}\hspace{.4pt}1145\discretionary{/}{%
}{/}1542362\hspace{.1pt}\discretionary{.}{%
}{.}\hspace{.4pt}1542407}}


\bibitem{Cohen-Steiner2007}
D.~Cohen{-}Steiner, H.~Edelsbrunner, and J.~Harer.
\newblock Stability of persistence diagrams.
\newblock {\em Discret. Comput. Geom.}, 37(1):103--120, 2007. doi: {{%
10\hspace{.1pt}\discretionary{.}{%
}{.}\hspace{.4pt}1007\discretionary{/}{%
}{/}s00454\discretionary{%
}{-}{-}006\discretionary{%
}{-}{-}1276\discretionary{%
}{-}{-}5}}


\bibitem{cplex2009v12}
I.~I. Cplex.
\newblock V12. 1: User’s manual for cplex.
\newblock {\em International Business Machines Corporation}, 46(53):157, 2009.

\bibitem{edelsbrunner09}
H.~Edelsbrunner and J.~Harer.
\newblock {\em Computational Topology - an Introduction}.
\newblock American Mathematical Society, 2010.

\bibitem{DBLP:conf/focs/EdelsbrunnerLZ00}
H.~Edelsbrunner, D.~Letscher, and A.~Zomorodian.
\newblock Topological persistence and simplification.
\newblock In {\em 41st Annual Symposium on Foundations of Computer Science,
  {FOCS} 2000, 12-14 November 2000, Redondo Beach, California, {USA}}, pp.
  454--463. {IEEE} Computer Society, 2000. doi: {{%
10\hspace{.1pt}\discretionary{.}{%
}{.}\hspace{.4pt}1109\discretionary{/}{%
}{/}SFCS\hspace{.1pt}\discretionary{.}{%
}{.}\hspace{.4pt}2000\hspace{.1pt}\discretionary{.}{%
}{.}\hspace{.4pt}892133}}


\bibitem{DBLP:journals/dcg/FabioL16}
B.~D. Fabio and C.~Landi.
\newblock The edit distance for reeb graphs of surfaces.
\newblock {\em Discret. Comput. Geom.}, 55(2):423--461, 2016. doi: {{%
10\hspace{.1pt}\discretionary{.}{%
}{.}\hspace{.4pt}1007\discretionary{/}{%
}{/}s00454\discretionary{%
}{-}{-}016\discretionary{%
}{-}{-}9758\discretionary{%
}{-}{-}6}}


\bibitem{DBLP:books/fm/GareyJ79}
M.~R. Garey and D.~S. Johnson.
\newblock {\em Computers and Intractability: {A} Guide to the Theory of
  NP-Completeness}.
\newblock W. H. Freeman, 1979.

\bibitem{intrinsicMTdistance}
E.~Gasparovic, E.~Munch, S.~Oudot, K.~Turner, B.~Wang, and Y.~Wang.
\newblock Intrinsic interleaving distance for merge trees.
\newblock {\em CoRR}, 1908.00063, 2019.

\bibitem{gurobi}
{Gurobi Optimization, LLC}.
\newblock {Gurobi Optimizer Reference Manual}, 2023.

\bibitem{heine16}
C.~Heine, H.~Leitte, M.~Hlawitschka, F.~Iuricich, L.~D. Floriani,
  G.~Scheuermann, H.~Hagen, and C.~Garth.
\newblock A survey of topology-based methods in visualization.
\newblock {\em Comput. Graph. Forum}, 35(3):643--667, 2016. doi: {{%
10\hspace{.1pt}\discretionary{.}{%
}{.}\hspace{.4pt}1111\discretionary{/}{%
}{/}cgf\hspace{.1pt}\discretionary{.}{%
}{.}\hspace{.4pt}12933}}


\bibitem{HilagaSKK01}
M.~Hilaga, Y.~Shinagawa, T.~Komura, and T.~L. Kunii.
\newblock {Topology matching for fully automatic similarity estimation of 3D
  shapes}.
\newblock In {\em ACM SIGGRAPH}, 2001.

\bibitem{DBLP:conf/cocoa/HongKY17}
E.~Hong, Y.~Kobayashi, and A.~Yamamoto.
\newblock Improved methods for computing distances between unordered trees
  using integer programming.
\newblock In X.~Gao, H.~Du, and M.~Han, eds., {\em Combinatorial Optimization
  and Applications - 11th International Conference, {COCOA} 2017, Shanghai,
  China, December 16-18, 2017, Proceedings, Part {II}}, vol. 10628 of {\em
  Lecture Notes in Computer Science}, pp. 45--60. Springer, 2017. doi: {{%
10\hspace{.1pt}\discretionary{.}{%
}{.}\hspace{.4pt}1007\discretionary{/}{%
}{/}978\discretionary{%
}{-}{-}3\discretionary{%
}{-}{-}319\discretionary{%
}{-}{-}71147\discretionary{%
}{-}{-}8\_4}}


\bibitem{DBLP:conf/dis/KondoOIY14}
S.~Kondo, K.~Otaki, M.~Ikeda, and A.~Yamamoto.
\newblock Fast computation of the tree edit distance between unordered trees
  using {IP} solvers.
\newblock In S.~Dzeroski, P.~Panov, D.~Kocev, and L.~Todorovski, eds., {\em
  Discovery Science - 17th International Conference, {DS} 2014, Bled, Slovenia,
  October 8-10, 2014. Proceedings}, vol. 8777 of {\em Lecture Notes in Computer
  Science}, pp. 156--167. Springer, 2014. doi: {{%
10\hspace{.1pt}\discretionary{.}{%
}{.}\hspace{.4pt}1007\discretionary{/}{%
}{/}978\discretionary{%
}{-}{-}3\discretionary{%
}{-}{-}319\discretionary{%
}{-}{-}11812\discretionary{%
}{-}{-}3\_14}}


\bibitem{DBLP:journals/cgf/LohfinkWLWG20}
A.~P. Lohfink, F.~Wetzels, J.~Lukasczyk, G.~H. Weber, and C.~Garth.
\newblock Fuzzy contour trees: Alignment and joint layout of multiple contour
  trees.
\newblock {\em Comput. Graph. Forum}, 39(3):343--355, 2020. doi: {{%
10\hspace{.1pt}\discretionary{.}{%
}{.}\hspace{.4pt}1111\discretionary{/}{%
}{/}cgf\hspace{.1pt}\discretionary{.}{%
}{.}\hspace{.4pt}13985}}


\bibitem{DBLP:conf/hvc/MantheyHB12}
N.~Manthey, M.~Heule, and A.~Biere.
\newblock Automated reencoding of boolean formulas.
\newblock In A.~Biere, A.~Nahir, and T.~E.~J. Vos, eds., {\em Hardware and
  Software: Verification and Testing - 8th International Haifa Verification
  Conference, {HVC} 2012, Haifa, Israel, November 6-8, 2012. Revised Selected
  Papers}, vol. 7857 of {\em Lecture Notes in Computer Science}, pp. 102--117.
  Springer, 2012.

\bibitem{morozov14}
D.~Morozov, K.~Beketayev, and G.~H. Weber.
\newblock Interleaving distance between merge trees.
\newblock In {\em TopoInVis}. 2014.

\bibitem{DBLP:conf/ppopp/MorozovW13}
D.~Morozov and G.~H. Weber.
\newblock Distributed merge trees.
\newblock In A.~Nicolau, X.~Shen, S.~P. Amarasinghe, and R.~W. Vuduc, eds.,
  {\em {ACM} {SIGPLAN} Symposium on Principles and Practice of Parallel
  Programming, PPoPP '13, Shenzhen, China, February 23-27, 2013}, pp. 93--102.
  {ACM}, 2013. doi: {{%
10\hspace{.1pt}\discretionary{.}{%
}{.}\hspace{.4pt}1145\discretionary{/}{%
}{/}2442516\hspace{.1pt}\discretionary{.}{%
}{.}\hspace{.4pt}2442526}}


\bibitem{DBLP:conf/apvis/NarayananTN15}
V.~Narayanan, D.~M. Thomas, and V.~Natarajan.
\newblock Distance between extremum graphs.
\newblock In S.~Liu, G.~Scheuermann, and S.~Takahashi, eds., {\em 2015 {IEEE}
  Pacific Visualization Symposium, PacificVis 2015, Hangzhou, China, April
  14-17, 2015}, pp. 263--270. {IEEE} Computer Society, 2015. doi: {{%
10\hspace{.1pt}\discretionary{.}{%
}{.}\hspace{.4pt}1109\discretionary{/}{%
}{/}PACIFICVIS\hspace{.1pt}\discretionary{.}{%
}{.}\hspace{.4pt}2015\hspace{.1pt}\discretionary{.}{%
}{.}\hspace{.4pt}7156386}}


\bibitem{DBLP:journals/tvcg/PontVDT22}
M.~Pont, J.~Vidal, J.~Delon, and J.~Tierny.
\newblock Wasserstein distances, geodesics and barycenters of merge trees.
\newblock {\em {IEEE} Trans. Vis. Comput. Graph.}, 28(1):291--301, 2022. doi:
  {{%
10\hspace{.1pt}\discretionary{.}{%
}{.}\hspace{.4pt}1109\discretionary{/}{%
}{/}TVCG\hspace{.1pt}\discretionary{.}{%
}{.}\hspace{.4pt}2021\hspace{.1pt}\discretionary{.}{%
}{.}\hspace{.4pt}3114839}}


\bibitem{DBLP:journals/cgf/SaikiaSW14}
H.~Saikia, H.~Seidel, and T.~Weinkauf.
\newblock Extended branch decomposition graphs: Structural comparison of scalar
  data.
\newblock {\em Comput. Graph. Forum}, 33(3):41--50, 2014. doi: {{%
10\hspace{.1pt}\discretionary{.}{%
}{.}\hspace{.4pt}1111\discretionary{/}{%
}{/}cgf\hspace{.1pt}\discretionary{.}{%
}{.}\hspace{.4pt}12360}}


\bibitem{DBLP:journals/tvcg/SridharamurthyM20}
R.~Sridharamurthy, T.~B. Masood, A.~Kamakshidasan, and V.~Natarajan.
\newblock Edit distance between merge trees.
\newblock {\em {IEEE} Trans. Vis. Comput. Graph.}, 26(3):1518--1531, 2020. doi:
  {{%
10\hspace{.1pt}\discretionary{.}{%
}{.}\hspace{.4pt}1109\discretionary{/}{%
}{/}TVCG\hspace{.1pt}\discretionary{.}{%
}{.}\hspace{.4pt}2018\hspace{.1pt}\discretionary{.}{%
}{.}\hspace{.4pt}2873612}}


\bibitem{DBLP:journals/tvcg/SridharamurthyN23}
R.~Sridharamurthy and V.~Natarajan.
\newblock Comparative analysis of merge trees using local tree edit distance.
\newblock {\em {IEEE} Trans. Vis. Comput. Graph.}, 29(2):1518--1530, 2023. doi:
  {{%
10\hspace{.1pt}\discretionary{.}{%
}{.}\hspace{.4pt}1109\discretionary{/}{%
}{/}TVCG\hspace{.1pt}\discretionary{.}{%
}{.}\hspace{.4pt}2021\hspace{.1pt}\discretionary{.}{%
}{.}\hspace{.4pt}3122176}}


\bibitem{DBLP:journals/jacm/Tai79}
K.~Tai.
\newblock The tree-to-tree correction problem.
\newblock {\em J. {ACM}}, 26(3):422--433, 1979. doi: {{%
10\hspace{.1pt}\discretionary{.}{%
}{.}\hspace{.4pt}1145\discretionary{/}{%
}{/}322139\hspace{.1pt}\discretionary{.}{%
}{.}\hspace{.4pt}322143}}


\bibitem{ThomasN13}
D.~M. Thomas and V.~Natarajan.
\newblock Detecting symmetry in scalar fields using augmented extremum graphs.
\newblock {\em {IEEE} Trans. Vis. Comput. Graph.}, 19(12):2663--2672, 2013.
  doi: {{%
10\hspace{.1pt}\discretionary{.}{%
}{.}\hspace{.4pt}1109\discretionary{/}{%
}{/}TVCG\hspace{.1pt}\discretionary{.}{%
}{.}\hspace{.4pt}2013\hspace{.1pt}\discretionary{.}{%
}{.}\hspace{.4pt}148}}


\bibitem{DBLP:journals/tvcg/TiernyFLGM18}
J.~Tierny, G.~Favelier, J.~A. Levine, C.~Gueunet, and M.~Michaux.
\newblock The topology toolkit.
\newblock {\em {IEEE} Trans. Vis. Comput. Graph.}, 24(1):832--842, 2018. doi:
  {{%
10\hspace{.1pt}\discretionary{.}{%
}{.}\hspace{.4pt}1109\discretionary{/}{%
}{/}TVCG\hspace{.1pt}\discretionary{.}{%
}{.}\hspace{.4pt}2017\hspace{.1pt}\discretionary{.}{%
}{.}\hspace{.4pt}2743938}}


\bibitem{wetzels2022path}
F.~Wetzels and C.~Garth.
\newblock A deformation-based edit distance for merge trees.
\newblock In {\em 2022 Topological Data Analysis and Visualization
  (TopoInVis)}, pp. 29--38, 2022. doi: {{%
10\hspace{.1pt}\discretionary{.}{%
}{.}\hspace{.4pt}1109\discretionary{/}{%
}{/}TopoInVis57755\hspace{.1pt}\discretionary{.}{%
}{.}\hspace{.4pt}2022\hspace{.1pt}\discretionary{.}{%
}{.}\hspace{.4pt}00010}}


\bibitem{repository}
F.~Wetzels, H.~Leitte, and C.~Garth.
\newblock Branch decomposition-independent edit distances (supplementary source
  code).
\newblock \url{https://github.com/scivislab/bdi-ed}, 2021.

\bibitem{wetzels2022branch}
F.~Wetzels, H.~Leitte, and C.~Garth.
\newblock Branch decomposition-independent edit distances for merge trees.
\newblock {\em Computer Graphics Forum}, 41(3):367--378, 2022. doi: {{%
10\hspace{.1pt}\discretionary{.}{%
}{.}\hspace{.4pt}1111\discretionary{/}{%
}{/}cgf\hspace{.1pt}\discretionary{.}{%
}{.}\hspace{.4pt}14547}}


\bibitem{Yan_geometry_aware}
L.~Yan, T.~Bin~Masood, F.~Rasheed, I.~Hotz, and B.~Wang.
\newblock Geometry aware merge tree comparisons for time-varying data with
  interleaving distances.
\newblock {\em IEEE Transactions on Visualization and Computer Graphics}, pp.
  1--1, 2022. doi: {{%
10\hspace{.1pt}\discretionary{.}{%
}{.}\hspace{.4pt}1109\discretionary{/}{%
}{/}TVCG\hspace{.1pt}\discretionary{.}{%
}{.}\hspace{.4pt}2022\hspace{.1pt}\discretionary{.}{%
}{.}\hspace{.4pt}3163349}}


\bibitem{surveyComparison2021}
L.~Yan, T.~B. Masood, R.~Sridharamurthy, F.~Rasheed, V.~Natarajan, I.~Hotz, and
  B.~Wang.
\newblock Scalar field comparison with topological descriptors: Properties and
  applications for scientific visualization.
\newblock {\em Comput. Graph. Forum}, 40(3):599--633, 2021. doi: {{%
10\hspace{.1pt}\discretionary{.}{%
}{.}\hspace{.4pt}1111\discretionary{/}{%
}{/}cgf\hspace{.1pt}\discretionary{.}{%
}{.}\hspace{.4pt}14331}}


\bibitem{YanWMGW20}
L.~Yan, Y.~Wang, E.~Munch, E.~Gasparovic, and B.~Wang.
\newblock A structural average of labeled merge trees for uncertainty
  visualization.
\newblock {\em {IEEE} Trans. Vis. Comput. Graph.}, 26(1):832--842, 2020. doi:
  {{%
10\hspace{.1pt}\discretionary{.}{%
}{.}\hspace{.4pt}1109\discretionary{/}{%
}{/}TVCG\hspace{.1pt}\discretionary{.}{%
}{.}\hspace{.4pt}2019\hspace{.1pt}\discretionary{.}{%
}{.}\hspace{.4pt}2934242}}


\bibitem{DBLP:journals/algorithmica/Zhang96}
K.~Zhang.
\newblock A constrained edit distance between unordered labeled trees.
\newblock {\em Algorithmica}, 15(3):205--222, 1996. doi: {{%
10\hspace{.1pt}\discretionary{.}{%
}{.}\hspace{.4pt}1007\discretionary{/}{%
}{/}BF01975866}}


\bibitem{DBLP:journals/ipl/ZhangSS92}
K.~Zhang, R.~Statman, and D.~E. Shasha.
\newblock On the editing distance between unordered labeled trees.
\newblock {\em Inf. Process. Lett.}, 42(3):133--139, 1992. doi: {{%
10\hspace{.1pt}\discretionary{.}{%
}{.}\hspace{.4pt}1016\discretionary{/}{%
}{/}0020\discretionary{%
}{-}{-}0190\discretionary{%
}{(}{(}92\discretionary{)}{%
}{)}90136\discretionary{%
}{-}{-}J}}


\end{thebibliography}
\end{document}